\documentclass[11pt]{article}

\usepackage{amsmath,amssymb}
\usepackage{amsthm}
\usepackage{color}
\usepackage{enumitem}
\usepackage{hyperref}
\usepackage{geometry}
\usepackage{pifont}
\usepackage{url}
\usepackage{graphicx}
\usepackage{placeins}
\usepackage{booktabs}
\usepackage{extarrows}

\graphicspath{{figures/}}

\numberwithin{equation}{section}
\theoremstyle{definition}
\newtheorem{theorem}{Theorem}[section]
\newtheorem{proposition}[theorem]{Proposition}
\newtheorem{corollary}[theorem]{Corollary}

\newtheorem{definition}[theorem]{Definition}
\newtheorem{example}[theorem]{Example}
\newtheorem{remark}[theorem]{Remark}

\usepackage[]{todonotes}

\reversemarginpar

\title{A Price-Based Framework for Stochastic Portfolio Theory}

\author{
	\textsc{Jongbong An} 
	\thanks{School of Mathematics, Statistics and Data Science, Sungshin Women’s University, South Korea (E-mail: {\it ajb8406@sungshin.ac.kr})}
	\and
	\textsc{Donghan Kim} 
	\thanks{Department of Mathematical Sciences, KAIST, South Korea (E-mail: {\it kimdonghan@kaist.ac.kr})}
}
\date{\today}

\begin{document}

\maketitle

\begin{abstract}
    \noindent
    We develop a price-based framework for stochastic portfolio theory in which trading strategies are generated from nominal price weights and evaluated relative to a price-weighted benchmark. Stock splits and reverse splits induce jumps in the generating weights without changing the value of existing investments. In a semimartingale market with predictable adjustment events, we incorporate the corresponding share adjustments into the self-financing condition and construct additively and multiplicatively generated strategies with explicit jump-corrected wealth decompositions. For additive generation, an entropy-based example exhibits relative wealth tending to $-\infty$ almost surely under repeated splits, demonstrating why the usual argument for long-horizon relative arbitrage does not extend directly. We also show that advance adjustment information alone provides no model-free guarantee of improved performance for functionally generated portfolios. For multiplicative generation, wealth-preserving restarts retain the classical portfolio allocation evaluated at the current price weights. We also establish a correspondence with the capitalization-based framework that preserves absolute wealth and yields a strategy-independent benchmark conversion. We illustrate the framework using daily NYSE data, comparing price- and capitalization-weighted benchmarks and the corresponding diversity-weighted portfolios across price- and capitalization-selected universes.
\end{abstract}

\bigskip

\section{Introduction} \label{sec: introduction}

Stochastic portfolio theory~(SPT), introduced by Fernholz \cite{fernholz2002stochastic}, studies equity market structure and portfolio performance; see also \cite{FK_survey}. In its classical formulation, portfolios are constructed from the vector of capitalization weights and evaluated relative to the capitalization-weighted market portfolio. A central tool is \emph{functional portfolio generation}, which associates a portfolio with a generating function of these weights. Fernholz's master formula \cite[(11.2)]{FK_survey} expresses the resulting relative wealth through changes in the generating function and an accumulated covariation term, without a stochastic integral. This decomposition connects observable features of the market, such as diversity and volatility, to sufficient conditions for relative arbitrage.

The theory has developed in several directions. Open markets and changes in the investment universe are studied in \cite{bayraktar2024quantifying, Open:market:Jacobi, Karatzas:Kim2}, while pathwise and model-free formulations are developed in \cite{rough:path:SPT, Karatzas:Kim, Schied:model-free}. Extensions of functional generation incorporate auxiliary processes and additional market information \cite{Market-to-book, rufxie2019generalised, strong2014generalizations}, and geometric interpretations relate generated portfolios to optimal transport \cite{palwong2016geometry}. Other developments concern universal portfolios \cite{Universal:SPT}, functional portfolio optimization \cite{Portfolio:Optimization:SPT}, robust asymptotic growth \cite{kardarasrobertson2021ergodic}, and signature-based portfolio methods \cite{cuchiero2025signature}.

Related work examines the volatility conditions and time horizons required for relative arbitrage \cite{vol:arb, laishkolnikovsoner2025relative, larssonruf2021relative}, the effects of transaction costs and price impact \cite{price:impact:SPT, Ruf:Xie}, and applications to commodity futures and option pricing \cite{fernholz2022permutation,Option:pricing:SPT}. Empirical studies investigate diversification, market structure, and portfolio performance \cite{banner2019diversification,campbell2025macroscopic}. Of particular relevance here is the \emph{additive generation} of trading strategies developed by Karatzas and Ruf \cite{karatzas2017trading}, in contrast to Fernholz's classical multiplicative construction. Its additive wealth decomposition yields particularly simple sufficient conditions for relative arbitrage: for a nonnegative concave generating function, sufficiently large accumulation of the associated nondecreasing covariation term guarantees outperformance with respect to the capitalization-weighted benchmark. 

Throughout much of this literature, capitalization weights remain the principal generating variables, and the capitalization-weighted market portfolio is the standard benchmark. Price weighting nevertheless underlies prominent equity indices, including the Dow Jones Industrial Average in the United States \cite{spdji_djia_methodology} and the Nikkei~225 in Japan \cite{nikkei225indexguidebook}. Although functional generation with alternative benchmarks has been studied \cite{strong2014generalizations}, using nominal price weights introduces the additional issue of changes in share units. This paper develops a price-based SPT framework in which the price weight vector generates trading strategies, and a price-weighted portfolio serves as the reference benchmark, explicitly incorporating stock splits and reverse splits into the construction.

The distinction from the capitalization-based setting is already visible at a single stock split. Let $P_i$ and $N_i$ denote the price and number of outstanding shares of stock~$i$, respectively, with capitalization $S_i=P_iN_i$. A split changes $P_i$ and $N_i$ in reciprocal proportions, leaving both $S_i$ and the value of an investor's existing position unchanged. It nevertheless changes, in general, the price weights $\nu_i=P_i/\sum_jP_j$, even though the capitalization weights $\mu_i=S_i/\sum_jS_j$ are unaffected. Consequently, a generating function evaluated at $\nu$ may jump without any corresponding change in investment wealth. A divisor adjustment can preserve the continuity of the price-weighted benchmark, but does not remove the jumps of the price weight vector. Functional generation must therefore accommodate a discontinuous generating state while preserving the self-financing property and continuity of wealth.

We work in a frictionless market with continuous, positive capitalization processes and locally constant outstanding share processes. Prices are thus right-continuous semimartingales with left limits, whose jumps occur at locally finitely many predictable adjustment events. The corresponding changes in share holdings are included directly in the self-financing condition. We introduce a divisor $\kappa$ so that the normalized total price $\widetilde\Sigma=(\sum_iP_i)/\kappa$ is continuous and represents the wealth of a self-financing price-weighted reference portfolio. This construction permits portfolio performance relative to $\widetilde\Sigma$ to be analyzed in terms of $\nu$, while distinguishing mechanical changes in share units from gains and losses due to market movements.

Within this framework, we first construct additively generated trading strategies from the price weights $\nu$, by adapting the methodology of Karatzas and Ruf \cite{karatzas2017trading}. For a twice continuously differentiable generating function $G$, the relative wealth of the resulting strategy $\varphi$ satisfies
\[
    V^\varphi(t;\nu) = G\big(\nu(t)\big) + \Gamma^{G,c}(t) - \sum_{0<s\le t}\Delta G\big(\nu(s)\big),
\]
where $\Gamma^{G,c}$ is the continuous covariation term, which is nondecreasing when $G$ is concave. The final sum removes changes in the generating function caused by adjustment events. This correction has an important consequence: boundedness of $G$ on the unit simplex no longer supplies a uniform lower bound for the generating-function value with its cumulative jumps removed. The usual argument based on the accumulation of $\Gamma^{G,c}$ therefore need not yield long-horizon relative arbitrage. In Example~\ref{ex: entropy market}, repeated splits reset a two-asset market to equal price weights, yet the relative wealth of the entropy-generated strategy tends to $-\infty$ almost surely for a suitable choice of model parameters. The example identifies the accumulated jump correction as an obstruction to applying the classical sufficient condition argument leading to relative arbitrage.

Stock splits and reverse splits are typically announced before their effective dates, raising the question of whether advance knowledge of the adjustment can overcome this difficulty. We compare strategies with and without access to the announcement information, starting from equal wealth at the announcement time. Their wealth difference up to the adjustment is an integral of the difference in their holdings against the continuous part of the price weights. The adjustment jump itself contributes no incremental wealth. Thus, any benefit from this information must arise through the continuous pre-adjustment trading exposure, rather than from the mechanical change in prices and share holdings. Without additional assumptions linking the announcement to the continuous market dynamics, the identity does not provide a model-free guarantee of outperformance.

Multiplicative generation requires a different treatment of adjustment events. A direct correction of self-financing and continuity defects need not preserve the portfolio allocation prescribed by the generating function. We instead restart the classical multiplicative construction between successive adjustment events and connect the resulting strategies through wealth-preserving normalizations. For a smooth positive generating function $G$, the resulting strategy $\psi$ satisfies a modified master formula
\[
    \log\frac{V^\psi(t;\nu)}{V^\psi(0;\nu)} = \log\frac{G\big(\nu(t)\big)}{G\big(\nu(0)\big)} + \int_0^t\frac{d\Gamma^{G,c}(s)}{G\big(\nu(s)\big)} - \sum_{0<s\le t}\Delta\log G\big(\nu(s)\big).
\]
The last term cancels the adjustment-induced jumps of $\log G(\nu)$, rather than representing a cash gain or loss at an event. At the same time, the portfolio allocation retains the classical functional form evaluated at the current price weights. In particular, the diversity-weighted generating function produces the familiar allocation $\pi_i=\nu_i^p/\sum_j\nu_j^p$ for $0<p<1$, including immediately after each adjustment.

Since the price- and capitalization-based frameworks evaluate performance relative to different benchmarks, a natural question is how to compare their trading strategies on a common basis. We establish a correspondence between trading strategies in the price- and capitalization-based frameworks that preserves absolute wealth. The corresponding relative wealth processes are related by a strategy-independent benchmark ratio.

We illustrate the frameworks using daily NYSE data from 2002 to 2021. We first compare the performance and diversification of price- and capitalization-weighted benchmarks on universes selected by stock price or market capitalization. Over the sample, price-weighted benchmarks have higher terminal performance, and price weights also generally exhibit higher cross-sectional entropy. These comparisons complement existing empirical work on price- and capitalization-weighted diversification \cite{damjanovic2025price}. We then examine the relative performance of $p$-diversity-weighted portfolios generated from the corresponding weights.

The diversity-weighted portfolios illustrate a further distinction between benchmark performance and the gains from functional generation. The reported price-based portfolios achieve smaller terminal gains over their own benchmarks than their capitalization-based counterparts, despite comparable accumulated excess growth terms. The wealth decomposition associates this difference with a larger decline in the adjusted state component under price weighting. When evaluated against the same capitalization-weighted benchmark, however, the reported price-based portfolios finish above their capitalization-based counterparts. We also quantify split-induced changes in the generating function, whose cumulative signs differ between the two universes.

The paper is organized as follows. Section~\ref{sec: setup} introduces the price market, trading strategies, and divisor. Section~\ref{sec: FGPTS} develops additive and multiplicative generation, examines the role of announced adjustment information, and establishes the connection with the capitalization-based framework. Section~\ref{sec: numerics} presents the empirical analysis. Proofs of the main construction results are collected in Appendices~\ref{app: proofs setup} and~\ref{app: proofs fgpts}.

\bigskip

\section{The Setup} \label{sec: setup}

This section introduces an equity market model and related definitions used throughout this paper.

\subsection{Price Market Model} \label{sec: price market model}

Consider a filtered probability space $(\Omega, \mathcal{F}, \mathbb{F}, \mathbb{P})$ where the filtration $\mathbb{F} = (\mathcal{F}(t))_{t \ge 0}$ satisfies the usual conditions (complete and right-continuous). All processes are assumed to be adapted to $\mathbb{F}$, and all relationships among random variables hold almost surely unless otherwise stated.

Let $S = (S_1, \dots, S_d)'$ be a $d$-dimensional continuous, positive semimartingale for $d \in \mathbb N$. Each positive component $S_i(t)$ represents the $i$-th asset's market capitalization at time $t\ge0$. The component sum
\begin{equation}    \label{def:Lambda}
    \Lambda(\cdot) := \sum_{i=1}^d S_i(\cdot)
\end{equation}
then represents the total market capitalization. The corresponding price process $P = (P_1, \dots, P_d)'$ is given by
\begin{equation}    \label{def:P}
    P_i(t) = \frac{S_i(t)}{N_i(t)}, \qquad i = 1,\dots,d, \quad t \ge 0,
\end{equation}
where $N_i(t)$ is the total number of outstanding shares of asset $i$ at time $t$. We assume that each $N_i$ is a locally constant, positive, and RCLL~(right-continuous with left limits) process, which may jump at discrete times due to corporate actions such as stock splits or reverse splits. Since each $S_i$ is strictly positive, every stock remains active in the market:
\[
    P_i(t) > 0, \qquad i = 1, \dots, d, \quad t \ge 0.
\]
We define an increasing sequence of stopping times to handle these jumps in a systematic way.

\begin{definition}[Sequence of adjustment and adjustment events] \label{def: adjustment sequence}
    A sequence of $\mathbb{F}$-stopping times $(\tau_k)_{k \ge 0}$ is called an \emph{adjustment sequence} for the outstanding share process $N = (N_1, \dots, N_d)'$, if
    \begin{enumerate}[label=(\roman*)]
        \item $\tau_0 = 0$, $\tau_k < \tau_{k+1}$ for all $k \ge 0$, and $\tau_k \to \infty$ as $k \to \infty$;
        \item every $\tau_k$ is predictable;
        \item $N$ is constant on each stochastic interval $[\tau_k, \tau_{k+1})$ for $k \ge 0$.
    \end{enumerate}
    For an adjustment sequence $(\tau_k)_{k \ge 0}$ for $N$, we say that $(\tau_k, m_k)_{k \ge 0}$ is a \emph{sequence of adjustment events}, if
    \begin{enumerate}[label=(\roman*)] \setcounter{enumi}{3}
        \item $m_0 = (1, \dots , 1)$ is a $d$-dimensional vector of ones;
        \item each $m_k$ is an $\mathcal{F}(\tau_k -)$-measurable, $\mathbb{R}^d$-valued, positive random vector for $k \ge 1$;
        \item the identity $N_i(\tau_k) = m_{k,i}N_i(\tau_k-)$ holds for every $k\ge1$ and $i=1,\dots,d$.
    \end{enumerate}
\end{definition}

A triple $(P, N, (\tau_k, m_k)_{k \ge 0})$ is called a \emph{price market}, where $P$ is the price process, $N$ is the outstanding share process, and $(\tau_k, m_k)_{k \ge 0}$ is a sequence of adjustment events for $N$. At adjustment event $(\tau_k, m_k)$, the number of outstanding shares $N_i$ jumps according to the identity in Definition~\ref{def: adjustment sequence} (vi). Equivalently, from the continuity of $S_i(t) = P_i(t)N_i(t)$ at $t=\tau_k$, $P_i$ is continuous on $[\tau_k, \tau_{k+1})$ and
\[
    P_i(\tau_k) = m_{k,i}^{-1} P_i(\tau_k -), \qquad i = 1,\dots,d, \quad k \ge 1.
\]
For example, $m_1 = (2, 1, 1, \dots, 1)$ means that only the first asset underwent a 2-for-1 split at the first adjustment date $\tau_1$ while others were fixed; the total number of outstanding shares of asset 1 doubled, whereas the price of asset 1 halved at $\tau_1$. Note that the market capitalization $P_1N_1$ of asset 1 is not affected by this split.

\begin{remark} \label{rem: adjustment sequence in real world}
    The predictability of $\tau_k$ and $\mathcal{F}(\tau_k-)$-measurability of $m_k$ requirements in Definition~\ref{def: adjustment sequence} are not restrictive in practice. Corporate actions such as stock splits or reverse splits are typically announced by the issuer several weeks or months before the \emph{actual adjustment date $\tau_k$}; the exact value of the multiplicative factor $m_k$, together with the date $\tau_k$ of adjustment, is announced and scheduled on the exchange calendar. From this perspective of market information, the events are not merely predictable but announced in advance.
\end{remark}

The market weight vector $\mu = (\mu_1,\dots,\mu_d)'$ in the capitalization-based SPT is defined as
\begin{equation}    \label{def:mu}
    \mu_i(t) := \frac{S_i(t)}{\Lambda(t)} = \frac{S_i(t)}{\sum_{j=1}^d S_j(t)}, \qquad i = 1,\dots,d, \quad t \ge 0.
\end{equation}
To develop the price-based SPT, we introduce the \emph{price weight vector} $\nu = (\nu_1,\dots,\nu_d)'$ with components
\begin{equation}    \label{def:nu}
    \nu_i(t) := \frac{P_i(t)}{\sum_{j=1}^d P_j(t)} = \frac{\mu_i(t)/N_i(t)}{\sum_{j=1}^d \mu_j(t)/N_j(t)} = \frac{S_i(t)/N_i(t)}{\sum_{j=1}^d S_j(t)/N_j(t)}, \qquad i = 1,\dots,d, \quad t \ge 0.
\end{equation}
Note that both $\mu$ and $\nu$ satisfy $\sum_{j=1}^d \mu_j(\cdot) \equiv 1 \equiv \sum_{j=1}^d \nu_j(\cdot)$. 

Despite its jumps, $\nu$ is predictable under the adjustment assumptions: its post-event value is determined by the pre-event prices and the $\mathcal F(\tau_k-)$-measurable multiplier $m_k$. This permits the generated portfolio holdings to depend on the current state $\nu(t)$, including at adjustment times, while stochastic gains across a jump are evaluated using the pre-event holdings. The price weight vector $\nu$ will serve as the fundamental state variable for the construction of price-based portfolios. We shall define the analogous version of \emph{trading strategy} and \emph{portfolio} in this price-based framework.

We begin with the heuristic meaning of the self-financing condition in this price-based framework. A trading strategy is described in terms of the number of shares held by an investor. At each adjustment event, the share count is rescaled by the corresponding split or reverse split factor. These jumps in share count must be accommodated within the stochastic calculus formulation; we therefore incorporate the jump terms directly into the self-financing condition so that the wealth process consistently reflects these adjustments. The market model in this paper is assumed to be frictionless, in the sense that rebalancing can be done continuously and does not incur any transaction costs.

\begin{definition} [Trading strategy] \label{def: trading strategy}
    Let $X$ be an RCLL semimartingale with locally finitely many jumps. A $d$-dimensional predictable process $\vartheta = (\vartheta_1,\dots,\vartheta_d)$ is called a \emph{trading strategy} with respect to $X$, if it is RCLL with locally finitely many jumps, integrable with respect to $X$, and if the corresponding value process
    \begin{equation}    \label{def:value-process}
        V^\vartheta(t; X) := \sum_{i=1}^d \vartheta_i(t) X_i(t), \qquad t \ge 0,
    \end{equation}
    is continuous in $t$ and satisfies the `self-financing' condition
    \begin{equation} \label{self-financing}
        V^\vartheta(t; X) = V^\vartheta(0; X) + \int_0^t \vartheta(s-) \cdot dX(s) + \sum_{0 < s \le t} \Delta \vartheta(s) \cdot X(s),  \qquad t \ge 0.
    \end{equation}
    Here and in what follows, $\Delta Y(s):=Y(s)-Y(s-)$ is the jump of an RCLL process $Y$ at time $s$.
\end{definition}

When $X$ is continuous (e.g., $X=S$ or $X=\mu$), continuity of $V^\vartheta$ implies $\Delta\vartheta(s)\cdot X(s)=0$, so \eqref{self-financing} reduces to the usual self-financing identity in \cite[Definition~2.1]{karatzas2017trading}. Thus, the definitions agree within the present class of predictable RCLL processes with locally finitely many jumps. In this paper, we mainly consider the case $X = P$ or $X = \nu$, where $X$ may have jumps at adjustment events. Hence, we need to involve the last jump terms in \eqref{self-financing} to accommodate possible jumps in $X$.

The self-financing condition indicates that neither withdrawals nor injections of capital are allowed; gains should be reinvested, and losses should be absorbed at all times. The predictability of $\tau_k$ and the $\mathcal F(\tau_k-)$-measurability of $m_k$ in Definition~\ref{def: adjustment sequence} are needed for the compensating adjustment to share holdings $\Delta\vartheta(\tau_k)$ to be chosen predictably and to preserve wealth through the corporate action event.

If we consider a \emph{price-based trading strategy} $\vartheta$ with respect to $P$, the $i$-th component $\vartheta_i(t)$ represents the number of shares of asset $i$ held at time $t$. Alongside $\vartheta$, it is convenient to introduce the corresponding \emph{portfolio} expressed as proportions of wealth.

\begin{definition} [Portfolio] \label{def: corresponding portfolio}
    Let $\vartheta$ be a trading strategy with respect to $X$ satisfying $V^\vartheta(\cdot; X) > 0$. The \emph{portfolio} $\pi = (\pi_1, \ldots, \pi_d)'$ corresponding to $\vartheta$ is defined by
    \begin{equation}    \label{def:pi}
        \pi_i(t) := \frac{\vartheta_i(t) X_i(t)}{V^\vartheta(t; X)}, \qquad i = 1, \ldots, d, \quad t \ge 0.
    \end{equation}
\end{definition}

Note that by definition the identity $\sum_{i=1}^d \pi_i(\cdot) \equiv 1$ holds for every portfolio.

In contrast to capitalization-weighted indices, a naive price-weighted index defined as the simple sum of constituent prices fails to maintain continuity in $t$, and therefore cannot serve as a consistent reference portfolio. The following example shows that a static trading strategy holding a single share per constituent violates the continuity condition in the presence of corporate actions. 

\begin{example}[Naive single-share strategy fails to be a trading strategy] \label{ex: naive failure}
    Consider a price market with $d = 2$ assets, a piecewise constant price process $P$ between adjustment events, and a single adjustment event $(\tau_1, m_1) = (1, (2, 1))$ corresponding to a $2$-for-$1$ split of asset~$1$:
    \[
        \big(P_1(t), P_2(t)\big) = (100, 50) \quad \text{for } t \in [0, \tau_1), \qquad \big(P_1(t), P_2(t)\big) = (50, 50) \quad \text{for } t \in [\tau_1, \infty).
    \]
    The static strategy $\vartheta_i \equiv 1$ (one share of each, no rebalancing) yields
    \[
        V^\vartheta(\tau_1 -; P) = 100 + 50 = 150, \qquad V^\vartheta(\tau_1; P) = 50 + 50 = 100,
    \]
    hence $V^\vartheta(\cdot;P)$ has a downward jump of $50$ at $\tau_1$. This violates the continuity requirement on $V^\vartheta$ in Definition~\ref{def: trading strategy}, even though the split is capitalization-preserving; each pre-split share is exchanged for two post-split shares of equal aggregate value, so the stockholder's market value is unchanged across $\tau_1$. The static strategy is therefore \emph{not} a self-financing trading strategy in our sense.

    To restore self-financeability, the share count of the splitting asset must be rescaled by $m_{1,1} = 2$ at $\tau_1$; that is, replace $\vartheta_1$ by $\widetilde\vartheta_1(t) := \mathbf{1}_{\{t < \tau_1\}} + 2 \times \mathbf{1}_{\{t \ge \tau_1\}}$ and keep $\widetilde\vartheta_2 \equiv 1$. Then the value process is continuous,
    \[
        V^{\widetilde\vartheta}(\tau_1) = 2 \times 50 + 1 \times 50 = 150 = V^{\widetilde\vartheta}(\tau_1-),
    \]
    and the explicit jump correction
    \[
        \Delta\widetilde\vartheta(\tau_1) \cdot P(\tau_1) = (1,0) \cdot (50, 50) = 50
    \]
    in~\eqref{self-financing} exactly compensates the integral contribution
    \[
        \widetilde\vartheta(\tau_1-) \cdot \Delta P(\tau_1) = (1,1) \cdot (-50, 0) = -50.
    \]
    This mechanical share adjustment preserves the value of the existing investment. Maintaining the price-weighted reference allocation requires an additional self-financing rebalance, as explained in Remark~\ref{rem: divisor interpretation}.
\end{example}

To resolve this structural inconsistency, the following subsection introduces the \emph{divisor}---a standard practice in real stock markets for handling split events---and provides a rigorous definition of the reference portfolio.

\medskip

\subsection{Divisor} \label{sec: divisor}
In what follows, let $\Sigma$ be the component sum of the price process, i.e., $\Sigma(\cdot) := \sum_{i=1}^d P_i(\cdot)$.

\begin{definition} [Divisor] \label{def: divisor}
    A predictable, RCLL, positive real-valued process $\kappa = (\kappa(t))_{t \ge 0}$ is called the \emph{divisor} for a price market $(P, N, (\tau_k, m_k)_{k \ge 0})$, if
    \begin{enumerate}[label=\textup{(\roman*)}, ref=\textup{(\roman*)}]
        \item $\kappa(0)=1$ and $\kappa$ is constant on each interval $[\tau_k, \tau_{k+1})$ for $k \ge 0$;
        \item at each adjustment event $(\tau_k, m_k)$ for $k \ge 1$, the divisor $\kappa$ jumps according to
        \begin{equation}    \label{item: divisor jump}
            \frac{\Sigma(\tau_k)}{\kappa(\tau_k)} = \frac{\Sigma(\tau_k-)}{\kappa(\tau_k-)}
        \end{equation}
        so that the process $\Sigma(\cdot)/\kappa(\cdot)$ is continuous.
    \end{enumerate}
\end{definition}

Here, the representation
\[
    \kappa(\tau_k) = \kappa(\tau_k -) \frac{\Sigma(\tau_k)}{\Sigma(\tau_k -)} = \kappa(\tau_k -)\frac{\sum_{i=1}^d m_{k,i}^{-1}P_i(\tau_k -)}{\sum_{i=1}^d P_i(\tau_k -)}, \qquad k \ge 1,
\]
shows that the divisor $\kappa$ is uniquely determined by a given price market, and is predictable and piecewise constant.

We call $\widetilde{\Sigma} := \Sigma/\kappa$ the \textit{normalized total price} process. The continuous process $\widetilde{\Sigma}$ can also be interpreted as the value process associated with the trading strategy $\vartheta(\cdot) = \mathbf{1}/\kappa(\cdot)$, where we denote $\mathbf{1}$ the $d$-dimensional vector of ones, as the following remark illustrates.

\begin{remark}[Financial interpretation of the divisor] \label{rem: divisor interpretation}
    The divisor $\kappa$ is not merely a mathematical artifact: it is a faithful continuous-time idealization of a long-standing market practice. The Dow Jones Industrial Average (DJIA) has been published since 1896 as a price-weighted index of leading US stocks; from 1928 onwards it has been computed in the form $\widetilde{\Sigma} = \Sigma(\cdot)/\kappa(\cdot)$ using the divisor mechanism documented in \cite{spdji_djia_methodology}. Each split of a constituent reduces $\Sigma$ at the event date and is absorbed by a multiplicative reduction of the divisor, exactly as in \eqref{item: divisor jump} of Definition~\ref{def: divisor}; the same mechanism handles reverse splits.

    From the perspective of an investor, $\kappa$ is the per-share-unit accounting adjustment that an index-tracking exchange-traded fund~(ETF) must apply each time a constituent restructures. A static portfolio holding one share of every constituent has value $\Sigma(t)$, which differs from the index $\widetilde\Sigma = \Sigma/\kappa$ by a multiplicative factor that jumps at every $\tau_k$. To track $\widetilde\Sigma$ continuously, the share count must instead be rescaled by $1/\kappa$: the strategy $\vartheta \equiv \mathbf{1}/\kappa$ is a self-financing replicating strategy for $\widetilde\Sigma$. At an adjustment event $\tau_k$ with $k \ge 1$, the mechanical share adjustment first changes these holdings to $m_{k,i}/\kappa(\tau_k-)$ in asset $i$. Rebalancing them to $1/\kappa(\tau_k)$ shares of each asset then preserves wealth by \eqref{item: divisor jump}, and therefore requires no additional capital.

    Forward splits reduce the divisor and reverse splits increase it, with the magnitude of each change determined by the adjustment multiplier and the affected stock's pre-event price weight. Among stocks selected at least once for either NYSE universe in Section~\ref{sec: numerics}, we identify $969$ forward splits and $118$ reverse splits over 2002--2021, a ratio of approximately $8.2:1$. This predominance is even more pronounced when the count is restricted to stocks belonging to the price-ranked universe on the event date: $753$ forward splits versus $12$ reverse splits. Figure~\ref{fig: divisor} shows the corresponding empirical split divisor $\kappa^P$, which declines from $1$ to approximately $0.41$ over the sample.
\end{remark}

As in \cite[Proposition~2.2]{karatzas2017trading}, Proposition~\ref{prop: equivalence} shows that $V^{\kappa\vartheta}(\cdot;\nu)$ is the wealth of the $P$-strategy $\vartheta$, expressed relative to the benchmark $\widetilde\Sigma$, the normalized price portfolio.

\begin{proposition} \label{prop: equivalence}
    An $\mathbb{R}^d$-valued predictable process $\vartheta = (\vartheta_1,\dots,\vartheta_d)$ is a trading strategy with respect to $P$ if and only if $\kappa \vartheta$ is a trading strategy with respect to $\nu$. The corresponding value processes satisfy
    \[
        V^{\kappa \vartheta}(\cdot; \nu) = \frac{V^\vartheta(\cdot; P)}{\widetilde{\Sigma}(\cdot)}.
    \]
\end{proposition}

The proof of Proposition~\ref{prop: equivalence}, which is a simple application of It\^o's product rule, is given in Appendix~\ref{app: proofs setup}. Proposition~\ref{prop: equivalence} reduces the comparison of trading strategies under $P$ to the analysis of strategies under $\nu$; the corresponding strategies are related by the scaling $\vartheta\mapsto\kappa\vartheta$, while relative wealth is obtained by dividing by $\widetilde\Sigma$. Hereafter, by a slight abuse of notation, we analyze the trading strategy $\vartheta$ directly under the process $\nu$.

\begin{definition} [Defects of self-financeability and continuity] \label{def: defect of self-financeability}
    For a predictable RCLL $\nu$-integrable vector process $\vartheta$ with locally finitely many jumps, we define
    \begin{equation} \label{eq: defect of self-financeability}
        Q^{\vartheta, \nu}(t) := \vartheta(t) \cdot \nu(t) - \vartheta(0) \cdot \nu(0) - \int_0^t \vartheta(s-) \cdot d\nu(s) - \sum\limits_{0 < s \le t} \Delta \vartheta(s) \cdot \nu(s), \qquad t \ge 0,
    \end{equation}
    and call it the \emph{defect of self-financeability} of $\vartheta$ with respect to $\nu$. We also define
    \begin{equation}    \label{def:defect-of-continuity}
        C^{\vartheta, \nu}(t) := \sum\limits_{0 < s \le t} \Delta (\vartheta \cdot \nu)(s), \qquad t \ge 0,
    \end{equation}
    and call it the \emph{defect of continuity} of $\vartheta$ with respect to $\nu$.
\end{definition}

The two defects serve different purposes. The process $Q^{\vartheta,\nu}$ is continuous, as follows directly from the jump identity for $\vartheta\cdot\nu$, whereas $C^{\vartheta,\nu}$ is the cumulative jumps of the marked-to-market value. The latter can be nonzero because of jumps in either $\nu$ or $\vartheta$. Since $\mathbf{1}\cdot\nu\equiv1$, subtracting the same scalar from every holding reduces the value by that scalar. This identity, together with the continuity of $Q^{\vartheta,\nu}$ and the pure jump structure of $C^{\vartheta,\nu}$, allows the two defects to be corrected separately.

For continuous $\vartheta$, the definition of $Q^{\vartheta,\nu}$ reduces to the usual self-financeability defect in \cite[(2.6)]{karatzas2017trading}, with $\nu$ in place of $\mu$. Both $Q^{\vartheta, \nu}$ and $C^{\vartheta, \nu}$ vanish if $\vartheta$ is a trading strategy. The following proposition is the price-based counterpart of \cite[Proposition~2.3]{karatzas2017trading}; its proof is given in Appendix~\ref{app: proofs setup}.

\begin{proposition} \label{prop: defect of self-financeability}
    Let $\vartheta$ be a predictable RCLL $\nu$-integrable vector process with locally finitely many jumps. Then, $\varphi$ with components
    \[
        \varphi_i(t) := \vartheta_i(t) - Q^{\vartheta, \nu}(t) - C^{\vartheta, \nu}(t), \qquad i=1,\dots, d, \quad t\ge 0
    \]
    is a trading strategy with respect to $\nu$.
\end{proposition}

We say that a trading strategy $\varphi$ with respect to $\nu$ is \emph{relative arbitrage} with respect to the (price-based) market, if there exists $T>0$ such that $V^{\varphi}(t;\nu) \ge 0$ for all $t \in [0, T]$ and
\[
    \mathbb{P}\big( V^{\varphi}(T;\nu) \ge V^{\varphi}(0;\nu) \big) = 1, \qquad \mathbb{P}\big( V^{\varphi}(T;\nu) > V^{\varphi}(0;\nu) \big) > 0.
\]
If the last probability is equal to $1$, then we say that $\varphi$ is \emph{strong relative arbitrage} with respect to the (price-based) market.

\bigskip

\section{Price-Based Functionally Generated Trading Strategies} \label{sec: FGPTS}

\subsection{Generating function and its gamma process}
We now develop the functional generation of trading strategies in the price-based framework, in parallel with the capitalization-based construction of \cite{karatzas2017trading}. We denote the unit simplex of $\mathbb{R}^d$ by
\[
    \triangle^d:=\Big\{x\in[0,\infty)^d:\sum_{i=1}^d x_i=1\Big\}.
\]
Under the standing positivity assumptions on $S$ and $N$, both the market weight vector $\mu$ and the price weight vector $\nu$ take values in $\triangle^d\cap(0,\infty)^d$. Accordingly, all generating functions below are defined on the open positive orthant $(0,\infty)^d$, and all their derivatives are evaluated there.

Let $X$ be a $d$-dimensional RCLL semimartingale with locally finitely many jumps, such that $X(t)\in(0,\infty)^d$ for every $t\ge0$ and $X(t-)\in(0,\infty)^d$ for every $t>0$. Let $G:(0,\infty)^d\to\mathbb R$ be of class $C^2$. We write
\[
    \nabla G(x):=\big(\partial_1G(x),\ldots,\partial_dG(x)\big),\qquad x\in(0,\infty)^d,
\]
and set $\vartheta(t):=\nabla G(X(t))$. We define the \emph{gamma process} $\Gamma^G$ associated with $G$ and $X$ as
\begin{equation} \label{eq: Gamma function}
    \Gamma^G(t; X) := G\big(X(0)\big) - G\big(X(t)\big) + \int_0^t \vartheta(s-) \cdot dX(s) + \sum\limits_{0 < s \le t} \Delta\vartheta(s) \cdot X(s), \qquad t \ge 0.
\end{equation}

\begin{definition} [Balanced function] \label{def: balanced}
    A function $G\in C^1((0,\infty)^d)$ is called \emph{balanced} if
    \[
        \nabla G(x)\cdot x=G(x),\qquad x \in (0,\infty)^d.
    \]
    We also write the defect of balance
    \begin{equation}    \label{def:defect-of-balance}
        b_G(x) := \nabla G(x) \cdot x - G(x).
    \end{equation}
\end{definition}

An application of It\^o's lemma shows that
\begin{equation}  \label{eq: decomposition of Gamma function}
    \Gamma^G(t) = \Gamma^{G,c}(t) + \Gamma^{G,d}(t), \qquad t \ge 0,
\end{equation}
where $\Gamma^{G,c}$ and $\Gamma^{G,d}$ are continuous part and jump part of $\Gamma^G$, respectively:
\begin{align*}
    \Gamma^{G,c}(t) := &-\frac{1}{2} \sum_{i,j=1}^d \int_0^t \partial^2_{ij} G\big(X(s-)\big) d[X_i, X_j]^c(s), \\
    \Gamma^{G,d}(t) := & \sum_{0 < s \le t} \Delta b_G\big(X(s)\big) =  \sum_{0 < s \le t} \Big( - \Delta G\big(X(s)\big) + \Delta(\vartheta \cdot X)(s) \Big).
\end{align*}
This shows that the process $\Gamma^G$ is of finite variation on each compact interval. Furthermore, if $G$ is balanced, then $b_G \equiv 0$ and $\Gamma^{G,d}\equiv0$, so $\Gamma^G$ is continuous.

Moreover, when $G$ is concave, the process $\Gamma^G$ satisfies
\[
\Gamma^G(t)\ge \sum_{0<s\le t}\Delta\vartheta(s)\cdot X(s), \qquad t\ge0.
\]
Indeed, the continuous part $\Gamma^{G,c}(t)$ is nondecreasing in $t$, hence nonnegative. The jump part can be expressed as
\begin{align*}
    \Gamma^{G,d}(t) &= \sum_{0 < s \le t} \Big( -\Delta G(X(s)) + \vartheta(s-) \cdot \Delta X(s) + \Delta\vartheta(s) \cdot X(s) \Big)
    \\
    &= \sum_{0 < s \le t}  \Big( - d_{B, G}\big(X(s), X(s-)\big) + \Delta\vartheta(s) \cdot X(s) \Big) \ge \sum_{0 < s \le t} \Delta\vartheta(s) \cdot X(s),
\end{align*}
where $d_{B, G}$ is the \emph{Bregman divergence} induced by $G$, defined as
\[
    d_{B, G}(x, y) := G(x) - G(y) - \nabla G(y) \cdot (x - y).
\]
The last inequality follows from the fact that the Bregman divergence is nonpositive if $G$ is concave.

In parallel with the decomposition \eqref{eq: decomposition of Gamma function} above, we introduce the continuous part of $G(X(\cdot))$:
\begin{equation} \label{eq: G^c definition}
    G^c\big(X(t)\big) := G\big(X(t)\big) - \sum_{0<s\le t} \Delta G\big(X(s)\big), \qquad t \ge 0.
\end{equation}
The notation $G^c(X(t))$ denotes the process obtained by removing the jumps from $G(X)$, not a deterministic function of the current state $X(t)$. It generally depends on the path up to time $t$ and retains the initial value $G(X(0))$. Here, ``continuous part'' includes continuous finite variation terms and does not mean only the continuous local martingale part.

\medskip

\subsection{Additive generation of price-based trading strategies}

In this subsection, we introduce additive generation of trading strategies, as the counterpart of \cite[Definition~4.1 and Proposition~4.3]{karatzas2017trading} in the price-based framework.

For a twice continuously differentiable function $G : (0, \infty)^d \to \mathbb{R}$, set a predictable, RCLL process $\vartheta(t) := \nabla G(\nu(t))$. By Proposition~\ref{prop: defect of self-financeability}, subtracting the defects $Q^{\vartheta,\nu}$ and $C^{\vartheta,\nu}$ from $\vartheta$ yields a trading strategy with respect to $\nu$. Subtracting, in addition, the defect of balance at $t=0$
\begin{equation}    \label{def:defect-of-balance-0}
    b_G\big(\nu(0)\big) = \nabla G\big(\nu(0)\big) \cdot \nu(0) - G\big(\nu(0)\big),
\end{equation}
preserves the self-financeability, since $\mathbf{1} \cdot d\nu \equiv 0$ allows a constant to be subtracted from each component. Using the identity
\[
    Q^{\vartheta,\nu}(t) + \Gamma^G(t) = \vartheta(t)\cdot\nu(t) - \vartheta(0)\cdot\nu(0) + G\big(\nu(0)\big) - G\big(\nu(t)\big), \qquad t \ge 0
\]
we obtain for $i=1, \dots, d$ and $t \ge 0$
\begin{equation} \label{eq: additively generated portfolio strategy}
    \begin{aligned}
        \varphi_i(t) &:= \vartheta_i(t) - Q^{\vartheta,\nu}(t) - C^{\vartheta,\nu}(t) - b_G\big(\nu(0)\big) \\
        &= \vartheta_i(t) - (\vartheta \cdot \nu)(t) - C^{\vartheta,\nu}(t) + G\big(\nu(t)\big) + \Gamma^G(t).
    \end{aligned}
\end{equation}

\begin{definition} [Additively generated trading strategy] \label{def: additively generated portfolio}
    The trading strategy $\varphi$ defined in \eqref{eq: additively generated portfolio strategy} is called the \emph{additively generated trading strategy} by $G$ with respect to $\nu$.
\end{definition}

The following result is the natural extension of \cite[Proposition~4.3]{karatzas2017trading}, with $G$ and $\Gamma^G$ replaced by their continuous parts $G^c$ in \eqref{eq: G^c definition} and $\Gamma^{G,c}$ in \eqref{eq: decomposition of Gamma function}, respectively; its proof is deferred to Appendix~\ref{app: proofs fgpts}.

\begin{proposition} [Additive generation] \label{prop: additively generated portfolio theorem}
    For $G \in C^2((0,\infty)^d)$, the value process in \eqref{def:value-process} of the additively generated trading strategy $\varphi$ by $G$ satisfies
    \begin{equation} \label{eq: additively generated portfolio}
        V^\varphi(\cdot;\nu) = G^c\big(\nu(\cdot)\big) + \Gamma^{G,c}(\cdot).
    \end{equation}
\end{proposition}

\begin{remark}[The event term in additive generation] \label{rem: jump sum irreducible}
    The counterpart of the identity \eqref{eq: additively generated portfolio} in the capitalization-based setting is given as
    \[
        V^\varphi(\cdot;\mu) = G\big(\mu(\cdot)\big) + \Gamma^{G}(\cdot)
    \]
    in \cite[Proposition~4.3]{karatzas2017trading}. For a nonnegative concave generating function, this additive structure yields a convenient sufficient condition for relative arbitrage when the gamma process accumulates beyond the initial generating-function value; see \cite[Theorem~5.1]{karatzas2017trading}. However, the term $G^c(\nu(t))$ in \eqref{eq: additively generated portfolio} carries the accumulated jump sum $\sum_{0<s\le t}\Delta G(\nu(s))$ as in \eqref{eq: G^c definition}, so it depends on the entire path up to time $t$; it can also be unbounded, even for bounded $G$, as the number of jumps can be unbounded. The previous capitalization-based SPT argument that a bounded, concave generating function $G$ lets its nondecreasing gamma process govern long-run relative wealth is therefore unavailable, and Example~\ref{ex: entropy market} below illustrates this breakdown, as the jump sum alone may drive the relative wealth to $-\infty$.
\end{remark}

Consider the Gibbs entropy function
\begin{equation}    \label{eq:Gibbs-entropy}
    H(x) = - \sum_{i=1}^d x_i \log x_i, \qquad x \in (0,\infty)^d.
\end{equation}
The function $H$ is smooth and concave on $(0,\infty)^d$. By Proposition~\ref{prop: additively generated portfolio theorem} with a direct calculation from \eqref{eq: Gamma function}, the value process of the additively generated trading strategy $\varphi$ by $H$ is computed as
\[
    V^\varphi(t;\nu) \equiv V^\varphi(t) = H^c \big(\nu(t)\big) + \frac{1}{2} \int_0^t \sum\limits_{i=1}^d \frac{d[\nu_i]^c(s)}{\nu_i(s-)}, \qquad t \ge 0.
\]
We construct a simple market model with a sequence of adjustment events such that the additively generated trading strategy $\varphi$ does not necessarily yield relative arbitrage with respect to the market portfolio over sufficiently long horizons. The example shows why the lower-bound argument underlying the capitalization-based entropy construction in \cite[Example~5.3]{karatzas2017trading} does not extend directly to price weights with repeated adjustments.

\begin{example} [Failure of entropy-generated relative arbitrage under repeated splits] \label{ex: entropy market}
    Fix a number $\delta \in (0, 1/2)$. Let $W_1$ and $W_2$ be two independent Brownian motions. Consider a market model with two assets whose price processes satisfy the following dynamics for positive constants $\alpha, \epsilon$: for each $k \ge 1$
    \begin{equation*}
        d\log P_1(t) = - \alpha\, dt + \epsilon\, dW_1(t), \qquad d\log P_2(t) = \alpha\, dt + \epsilon\, dW_2(t), \qquad \tau_{k-1} < t < \tau_k
    \end{equation*}
    with $P_1(0) = P_2(0) = 1$ where $(\tau_k, m_k)_{k \ge 0}$ is a sequence of adjustment events such that
    \[
        \tau_k := \inf \bigg\{ t > \tau_{k-1} : \nu_1(t-) = \frac{P_1(t-)}{P_1(t-)+P_2(t-)} = \delta \bigg\}, \qquad m_k = \bigg(1, \, \frac{1-\delta}{\delta}\bigg), \qquad k \ge 1.
    \]
    Thus $\nu_1(\tau_k-)=\delta$, whereas the adjustment resets the post-event price weights to $\nu(\tau_k)=(1/2, 1/2)$, hence $P_1(\tau_k) = P_2(\tau_k)$. Here, each $\tau_k$ is predictable as it is a hitting time of a continuous process. The values of $\alpha$ and $\epsilon$ will be specified below.

    In this market model, two assets play opposite roles. The price weight $\nu_1$ of the first asset drifts downward from its initial value $\nu_1(0)=1/2$, and the adjustment time $\tau_k$ is triggered precisely when $\nu_1$ falls to the floor level $\delta$. At each $\tau_k$, the policy equalizes the two prices through a value-preserving split of the second asset $P_2$, resetting the price weights to the initial state $(1/2, 1/2)$. Indeed, a company whose share price has risen far above its peers tends to split its stock, lowering the per-share price to restore the liquidity of its shares. This parallels, for example, the split-merge rule of \cite[Section~3.4]{strong2011diversity}, where the largest company is split in capitalization at the diversity boundary; our split acts instead on the price, leaving capitalizations unchanged.

    A direct computation gives the continuous quadratic variations of $\log\nu_1$ and $\log\nu_2$:
    \[
        d[\log \nu_1]^c(t) = 2\epsilon^2\big(1-\nu_1(t)\big)^2\,dt, \qquad
        d[\log \nu_2]^c(t) = 2\epsilon^2\big(1-\nu_2(t)\big)^2\,dt, \qquad t\ge0.
    \]
    The factor $2\epsilon^2$ originates from the difference of the two independent Brownian drivers. Writing $\log\nu_1 = \log P_1 - \log\Sigma$ with $\Sigma = P_1+P_2$, one verifies that the diffusion part of $d\log\nu_1$ equals $\nu_2(t)\,\epsilon\,(dW_1(t) - dW_2(t))$, since $\log P_1$ contributes $\epsilon\,dW_1$ while $\log\Sigma$ contributes $\nu_1\epsilon\,dW_1 + \nu_2\epsilon\,dW_2$ in its diffusion. Independence of $W_1,W_2$ then yields $d[W_1 - W_2] = 2\,dt$, hence $d[\log\nu_1]^c = 2(\nu_2\epsilon)^2\,dt$. Equivalently, $W_2 - W_1 = \sqrt{2}\,B$ for a standard Brownian motion $B$, which is consistent with the volatility $\sqrt{2}\epsilon$ of $Y(t) := \log(\nu_2(t)/\nu_1(t))$ used in the hitting-time representation below.

    Since each cycle begins with the two prices equalized, the process $Y(t) = \log(\nu_2(t)/\nu_1(t))$ restarts from $0$ at every $\tau_{k-1}$ and evolves as a Brownian motion with drift $2\alpha$ and volatility $\sqrt{2}\epsilon$ on each inter-event interval. The increment $\tau_k - \tau_{k-1}$ is therefore the first-passage time of this drifted Brownian motion to the level $\log\big((1-\delta)/\delta\big)$:
    \[
        \tau_k - \tau_{k-1} \stackrel{d}{=} \inf\Big\{ t > 0 : 1-\delta = \delta e^{2\alpha t + \sqrt{2}\epsilon B(t)} \Big\}.
    \]
    By the Brownian strong Markov property, the increments
    \[
        X_k := \tau_k - \tau_{k-1}, \qquad k \ge 1,
    \]
    are i.i.d. and finite with mean
    \[
        m := \mathbb{E}[X_1] = \frac{1}{2\alpha}\log\frac{1-\delta}{\delta} \in (0, \infty),
    \]
    which follows by differentiating at $\lambda = 0$ the Laplace transform of the first-passage time given in \cite[Section~3.5.C]{karatzas1991brownian}. In particular, $\tau_k \to \infty$ as $k \to \infty$, almost surely. We write
    \[
        \begin{aligned}
        V^{\varphi}(t) - H\big(\nu(0)\big) &= \sum_{j=1}^{\infty} \bigg( H\big(\nu(\tau_j-)\big) - H\big(\nu(\tau_{j-1})\big) + \epsilon^2 \int_{\tau_{j-1}}^{\tau_j} f\big(\nu(s)\big) \, ds \bigg)\mathbf{1}_{\{\tau_j \le t\}} \\
        & \qquad + \sum_{k=1}^{\infty} \bigg( H\big(\nu(t)\big) - H\big(\nu(\tau_{k-1})\big) + \epsilon^2 \int_{\tau_{k-1}}^{t} f\big(\nu(s)\big) \, ds \bigg)\mathbf{1}_{\{\tau_{k-1} \le t < \tau_k\}},
        \end{aligned}
    \]
    where
    \[
        f\big(\nu(s)\big) := \nu_1(s)\big(1-\nu_1(s)\big)^2 + \nu_2(s)\big(1-\nu_2(s)\big)^2 = \nu_1(s) \nu_2(s) \in \left[0,\scalebox{0.85}{$\dfrac14$}\right].
    \]
    Let $K(t):=\max(\{j:\tau_j \le t\} \cup \{0\})$. Since $H(\nu(t)) - H(\nu(\tau_{k-1})) \le 0$ for any $k \ge 1$, the second sum is bounded pathwise by $\epsilon^2(t - \tau_{K(t)})/4 \le \epsilon^2 t/4$. For the first sum, the same entropy drop $b := H(\nu(0)) - H(\nu(\tau_1-)) = \log 2 - H(\delta, 1-\delta) > 0$ occurs on each completed interval, and $f \le 1/4$ together with $\sum_{j=1}^{K(t)} X_j = \tau_{K(t)} \le t$ bounds it by $-bK(t) + \epsilon^2 t/4$. Combining the two parts, we obtain the pathwise bound
    \[
        V^{\varphi}(t) - H\big(\nu(0)\big) \;\le\; -b K(t) + \frac{\epsilon^2}{2}\,t, \qquad t \ge 0.
    \]
    Since the inter-event times $(X_k)_{k\ge1}$ are positive i.i.d.\ random variables with mean $m$, the renewal strong law yields $K(t)/t \to 1/m$ almost surely. Consequently,
    \[
        \limsup_{t\to\infty}\frac{V^\varphi(t)}{t} \;\le\; -\frac{b}{m}+\frac{\epsilon^2}{2} \qquad\text{a.s.}
    \]
    If $\alpha$ is sufficiently large, we can ensure that the last term is negative, hence $V^{\varphi}(t) \to -\infty$ as $t \to \infty$, almost surely. Thus, the entropy-generated strategy cannot remain admissible over arbitrarily long horizons.
    
    In the capitalization-based setting, entropy-generated relative arbitrage requires sufficient accumulation of the gamma process, as specified in \cite[Example~5.3]{karatzas2017trading}. However, this lower-bound argument used for continuous capitalization weights does not extend directly to price weights with repeated adjustments.
\end{example}

\medskip

\subsection{The role of announced adjustment information}
\label{sec: role of announcement information}

Example~\ref{ex: entropy market} shows that the strategy additively generated by entropy can fail to yield relative arbitrage under repeated share adjustments. This raises the question of whether advance knowledge of adjustment times and multipliers can help an investor overcome this failure by modifying the strategy. To examine the role of this additional information, we compare trading strategies with and without access to an adjustment announcement.

Let $\mathbb F^0=(\mathcal F^0(t))_{t\ge0}$ denote the filtration generated by the price process $P$ in \eqref{def:P} without the announcement information, and let $\mathbb F=(\mathcal F(t))_{t\ge0}$ be the enlarged filtration, introduced in Section~\ref{sec: setup}, satisfying $\mathcal F^0(t)\subseteq\mathcal F(t)$ for every $t\ge0$. Fix an adjustment event $(\tau,m)$ announced at a predictable stopping time $\sigma<\tau$, so that $(\tau,m)$ is $\mathcal F(\sigma)$-measurable, and suppose that no adjustment occurs on $[\sigma,\tau)$. The informed strategy remains predictable; in particular, information first revealed at $\sigma$ may affect its holdings at times strictly after $\sigma$, but cannot be used to prescribe a new jump in holdings at $\sigma$ itself. Recall that $\nu$ is an RCLL $\mathbb{F}$-semimartingale with locally finitely many jumps, and we write $\nu^c:=\nu-\sum_{0<s\le\cdot}\Delta\nu(s)$ for its continuous part.

The distinction between $\mathbb F^0$ and $\mathbb F$ is not needed for the algebraic identity below, which follows from the self-financing property and the continuity of the value processes. Its role is to distinguish the information available to the two strategies. The holdings of $\vartheta$ before the adjustment do not use the announcement, whereas $\eta$ may alter its pre-adjustment holdings after observing $(\tau,m)$. Thus, $\eta-\vartheta$ represents the incremental trading exposure made possible by the announcement information.

\begin{proposition}[Continuous source of announcement-based performance]    \label{prop: continuous source of incremental wealth}
    Let $\vartheta$ and $\eta$ be two trading strategies with respect to $\nu$. Suppose that the holdings of $\vartheta$ on $[\sigma,\tau)$ are $\mathbb F^0$-predictable, whereas $\eta$ may use the announcement information contained in $\mathbb F$, and assume that $V^\eta(\sigma;\nu)=V^\vartheta(\sigma;\nu)$. Then
    \begin{equation}    \label{eq: announcement period wealth comparison}
        V^\eta(t;\nu)-V^\vartheta(t;\nu) = \int_\sigma^t \big(\eta(s-)-\vartheta(s-)\big)\cdot d\nu^c(s), \qquad \sigma \le t \le\tau.
    \end{equation}
    In particular, the adjustment jump at $\tau$ produces no instantaneous change in the wealth difference between the informed and uninformed strategies.
\end{proposition}

\begin{proof}
    Set $\delta:=\eta-\vartheta$. Subtracting the two self-financing identities \eqref{self-financing} and using the equality of the value processes at time $\sigma$ gives
    \[
        V^\eta(t;\nu)-V^\vartheta(t;\nu) = \int_\sigma^t\delta(s-)\cdot d\nu(s) + \sum_{\sigma<s\le t}\Delta\delta(s)\cdot\nu(s).
    \]
    Since $\nu$ has locally finitely many jumps,
    \[
        \int_\sigma^t\delta(s-)\cdot d\nu(s) = \int_\sigma^t\delta(s-)\cdot d\nu^c(s) + \sum_{\sigma<s\le t}\delta(s-)\cdot\Delta\nu(s).
    \]
    At every jump time $s$, the continuity of both value processes implies
    \[
        0 = \Delta\big(V^\eta-V^\vartheta\big)(s) = \delta(s-)\cdot\Delta\nu(s) + \Delta\delta(s)\cdot\nu(s).
    \]
    Thus all jump contributions cancel, and \eqref{eq: announcement period wealth comparison} follows.
\end{proof}

\begin{remark}[No model-free relative arbitrage guarantee from the announcement alone] \label{rem: no model free guarantee from announcement}
    Equation~\eqref{eq: announcement period wealth comparison} shows that the announcement may affect relative performance only through the pre-event continuous trading exposure chosen by the informed strategy. The mechanical adjustment jump at event time $\tau$ itself contributes no incremental wealth. The stochastic integral on the right-hand side, however, has no definite sign under the present assumptions. In particular, without additional conditions relating the announcement information to the continuous evolution of $\nu$, one cannot infer the existence of an $\mathbb F$-predictable strategy $\eta$ satisfying
    \[
        V^\eta(\tau;\nu)\ge V^\vartheta(\tau;\nu)\quad\text{a.s.}, \qquad \mathbb P\big(V^\eta(\tau;\nu) > V^\vartheta(\tau;\nu)\big) > 0.
    \]
    Thus, the announced adjustment time and multiplier alone do not provide a model-free sufficient condition for the informed strategy to outperform the uninformed one. This observation does not rule out informed outperformance in a specified market; any such positive result must additionally exploit structural properties of the continuous pre-event dynamics.
\end{remark}

\subsection{Multiplicative generation of price-based trading strategies}

We now return to functional generation and consider the multiplicative construction. In contrast to additive generation, Fernholz's classical multiplicative generation of portfolios requires extra attention. One could follow the additive methodology as in the multiplicative construction of \cite[Definition~4.7]{karatzas2017trading}; namely, setting
\[
    \eta_i(\cdot) := \partial_i G \big(\nu(\cdot)\big)\, \exp\bigg( \int_{0}^{\cdot} \frac{d\Gamma^{G, c}(s)}{G(\nu(s))} \bigg)
\]
globally and subtracting $Q^{\eta,\nu}$, $C^{\eta,\nu}$, and the initial balance defect $b_G(\nu(0))$ from each component of $\eta$ would again produce a trading strategy, now with initial wealth $G(\nu(0))$. However, even when its wealth remains positive, the resulting portfolio need not retain Fernholz's classical multiplicative allocation after an adjustment event.

For example, the diversity-weighted function $\mathbf{D}_p$ of Example~\ref{ex: diversity weighted} below, applied with the continuous market weight $\mu$, is known to multiplicatively generate so-called $p$-diversity-weighted portfolio
\[
    \pi^{\mu}_i(t) = \frac{\mu_i^p(t)}{\sum_{j=1}^d \mu_j^p(t)}, \qquad i = 1, \ldots, d, \quad t \ge 0.
\]
However, when one tries to replace $\mu$ with $\nu$, the discontinuity of $\nu$ at each $\tau_k$ influences the individual portfolio weights differently for $i = 1, \ldots, d$, failing to maintain the $p$-diversity-weighted wealth distribution after each adjustment event.

To remedy this, we therefore restart the classical multiplicative construction on each inter-event interval and rebalance by a factor at each $\tau_k$ in a self-financing way to maintain the desired wealth distribution. We shall illustrate this `gluing' procedure in more detail in the following.

Let $G : (0,\infty)^d \to (0, \infty)$ be twice continuously differentiable with $1/G$ locally bounded. On each inter-event interval $t \in [\tau_k, \tau_{k+1})$ for $k \ge 0$, define
\begin{equation} \label{eq: tps eta}
    \eta^{(k)}_i(t) := \partial_i G \big(\nu(t)\big)\, \exp\bigg( \int_{\tau_k}^{t} \frac{d\Gamma^{G, c}(s)}{G\big(\nu(s)\big)} \bigg),
\end{equation}
the defect of balance constant
\begin{equation} \label{eq: tps balance defect}
    b_G\big(\nu(\tau_k)\big) := \nabla G\big(\nu(\tau_k)\big) \cdot \nu(\tau_k) - G\big(\nu(\tau_k)\big),
\end{equation}
and the inter-event defect of self-financeability
\begin{equation} \label{eq: tps Q}
    Q^{(k)}(t) := (\eta^{(k)} \cdot \nu)(t) - (\eta^{(k)} \cdot \nu)(\tau_k) - \int_{\tau_k}^{t} \eta^{(k)}(s) \cdot d\nu(s),
\end{equation}
obtained by specializing~\eqref{eq: defect of self-financeability} to the interval $[\tau_k, t]$. Note that the jump term of~\eqref{eq: defect of self-financeability} vanishes here, since $\eta^{(k)}$ and $\nu$ are continuous on the inter-event interval. Define $\psi$ with components
\begin{equation} \label{eq: multiplicatively generated trading strategy}
    \psi_i(\cdot) := \alpha_k \Big(\eta^{(k)}_i(\cdot) - Q^{(k)}(\cdot) - b_G\big(\nu(\tau_k)\big) \Big) \quad \text{on } [\tau_k, \tau_{k+1}) \text{ for } k \ge 0,
\end{equation}
where the normalization constants $(\alpha_k)_{k \ge 0}$ are specified inductively by
\begin{equation} \label{eq: tps alpha}
    \alpha_0 := 1, \qquad
    \alpha_k := \frac{V^\psi(\tau_k -;\nu)}{G\big(\nu(\tau_k)\big)} \quad \text{for } k \ge 1.
\end{equation}
For each $k \ge 1$, $\alpha_k$ is $\mathcal F(\tau_k-)$-measurable, since both $V^\psi(\tau_k-;\nu)$ and $\nu(\tau_k)$ are $\mathcal F(\tau_k-)$-measurable.

At each restart, the local construction is evaluated at the post-adjustment state $\nu(\tau_k)$ and has initial value $G(\nu(\tau_k))$. The factor $\alpha_k$ scales this value to the wealth carried over from the preceding interval. Thus, the restart restores the prescribed allocation without adding or withdrawing capital; the normalization changes the scale of the holdings, not their portfolio proportions.

The following result verifies that the vector $\psi$ is a trading strategy with respect to $\nu$, and Proposition~\ref{prop: multiplicatively generated trading strategy theorem} below is the natural extension of \cite[Proposition~4.8]{karatzas2017trading} to the presence of adjustment events. Their proofs are given in Appendix~\ref{app: proofs fgpts}.

\begin{proposition}[Multiplicatively generated trading strategy] \label{prop: multiplicatively generated trading strategy}
    The process $\psi$ with components given in \eqref{eq: multiplicatively generated trading strategy} is a trading strategy with respect to $\nu$.
\end{proposition}

\begin{definition}[Multiplicatively generated trading strategy] \label{def: multiplicatively generated trading strategy}
    The trading strategy $\psi$ of \eqref{eq: multiplicatively generated trading strategy} is called the \emph{multiplicatively generated trading strategy} by $G$ with respect to $\nu$.
\end{definition}

\begin{proposition}[Multiplicative generation] \label{prop: multiplicatively generated trading strategy theorem}
    Suppose that $G : (0, \infty)^d \to (0, \infty)$ is twice continuously differentiable with $1/G$ locally bounded. The log-value process of the multiplicatively generated trading strategy $\psi$ by $G$ satisfies
    \begin{equation} \label{eq: value process of multiplicatively generated trading strategy}
        \log V^\psi(\cdot;\nu) = (\log G)^c \big(\nu(\cdot)\big) + \int_0^{\cdot} \frac{d\Gamma^{G, c}(s)}{G\big(\nu(s)\big)},
    \end{equation}
    where, in parallel with~\eqref{eq: G^c definition} but applied to $\log G$,
    \[
        (\log G)^c\big(\nu(t)\big) := \log G \big(\nu(t)\big) - \sum_{0 < s \le t} \Delta\log G\big(\nu(s)\big), \qquad t \ge 0
    \]
    is the continuous part of $\log G(\nu(\cdot))$.
\end{proposition}

\begin{corollary}[Multiplicatively generated portfolio] \label{coro: representation of multiplicatively generated trading strategy balanced case}
    The portfolio corresponding to $\psi$ in Definition~\ref{def: multiplicatively generated trading strategy} admits the representation
    \begin{equation} \label{eq: tps related portfolio}
        \pi^\psi_i(t) = \nu_i(t)\Bigg( 1 + \frac{\partial_i G\big(\nu(t)\big) - \sum_{j=1}^d \nu_j(t)\,\partial_j G\big(\nu(t)\big)}{G\big(\nu(t)\big)} \Bigg), \qquad i = 1, \dots, d, \quad t \ge 0,
    \end{equation}
    which coincides with the capitalization-based setting of \cite[Eq (4.8)]{karatzas2017trading}, evaluated at $\nu(t)$. If $G$ is balanced as in Definition~\ref{def: balanced}, then the representation \eqref{eq: tps related portfolio} reduces to
    \begin{equation*}
        \pi^\psi_i(t) = \nu_i(t)\,\partial_i \log G\big(\nu(t)\big), \qquad i = 1, \dots, d, \quad t \ge 0.
    \end{equation*}
\end{corollary}

\begin{remark}[Reduction to the continuous case] \label{rem: tps continuous reduction}
    If $N$ is constant, all adjustment multipliers equal $\mathbf{1}$ and the price weights $\nu$ are continuous. An adjustment sequence satisfying Definition~\ref{def: adjustment sequence} may still be retained as a sequence of auxiliary restart times. The wealth-preserving restarts do not change the globally generated strategy, so the construction coincides with the classical multiplicative generation of \cite[Proposition~4.8]{karatzas2017trading}. In particular, $(\log G)^c=\log G$, and \eqref{eq: value process of multiplicatively generated trading strategy} reduces to Fernholz's classical master formula \cite[(11.2)]{FK_survey}.
\end{remark}

\begin{example}[Diversity-weighted portfolios]\label{ex: diversity weighted}
    Consider the diversity-weighted function for $0 < p < 1$
    \[
        \mathbf{D}_p(x) := \bigg( \sum_{i=1}^d x_i^p \bigg)^{\frac 1p}, \qquad x \in (0,\infty)^d.
    \]
    The function $\mathbf{D}_p$ is smooth, positive, concave, and balanced in the sense of Definition~\ref{def: balanced}, and $1/\mathbf{D}_p(\nu)$ is locally bounded under the standing positivity assumption of $\nu$ stated in Section~\ref{sec: price market model}. Definition~\ref{def: multiplicatively generated trading strategy} therefore applies: the multiplicatively generated trading strategy $\psi$ by $\mathbf{D}_p$ is well-defined, and Corollary~\ref{coro: representation of multiplicatively generated trading strategy balanced case} reduces its portfolio to Fernholz's $p$-diversity weights of \cite[Example~3.4.4]{fernholz2002stochastic} in the price-based framework
    \begin{equation*}
        \pi^{\psi, (p)}_i(\cdot) = \frac{\nu_i^p(\cdot)}{\sum_{j=1}^d \nu_j^p(\cdot)}, \qquad i = 1, \dots, d.
    \end{equation*}
    We again note that these portfolio weights depend only on the current state $\nu(t)$ and carry no adjustment history dependence. Proposition~\ref{prop: multiplicatively generated trading strategy theorem} specializes to
    \begin{equation*}
        \log V^\psi(t;\nu) = (\log \mathbf{D}_p)^c \big(\nu(t)\big) + \int_0^t \frac{d\Gamma^{\mathbf{D}_p, c}(s)}{\mathbf{D}_p \big(\nu(s)\big)}, \qquad t \ge 0,
    \end{equation*}
    with the wealth-preserving rebalancing at each adjustment event $\tau_k$ ensuring that the portfolio remains at $\pi^{\psi, (p)}(\tau_k)$.
\end{example}

\medskip

\subsection{Comparison with the capitalization-based framework} \label{sec: comparison}

As we study two different choices of benchmarks, capitalization-weighted and price-weighted indices, it is natural to ask how trading strategies of the price-based framework compare with those of the capitalization-based one. The two frameworks turn out to be linked by an explicit correspondence of strategies, and the induced benchmark conversion does not depend on the strategy. To make this precise, we examine the trading strategies with respect to $P$ and those with respect to $S$.

Suppose that an investor holds $a(t) := (a_1(t), \dots, a_d(t))'$ shares of the constituent assets at time $t \ge 0$. The value of this investor's holdings at time $t$ is then given by
\[
    a(t) \cdot P(t) = \sum_{i=1}^d a_i(t) P_i(t) = \sum\limits_{i=1}^d \frac{a_i(t)}{N_i(t)} S_i(t).
\]
This identity shows that the `strategy' $a(t)$ with respect to $P$ corresponds to the `strategy' $a(t) \odot N(t)^{-1}$ with respect to $S$; here and in what follows, $A^{-1}$ and $A \odot B$ for $d$-dimensional vectors $A$ and $B$ represent the component-wise inverse and the component-wise product, respectively. We formalize this correspondence in the following result, and prove it in Appendix~\ref{app: proofs fgpts}.

\begin{proposition}[Change of basis] \label{prop: change of basis}
    For an $\mathbb{R}^d$-valued predictable process $\vartheta$, the following conditions are equivalent:
    \begin{enumerate}
        \item[\textup{(i)}] $\vartheta$ is a trading strategy with respect to $P$.
        \item[\textup{(ii)}] $\vartheta \odot N^{-1}$ is a trading strategy with respect to $S$.
    \end{enumerate}
    When (i)--(ii) hold, their value processes coincide:
    \begin{equation} \label{eq: change of basis wealth}
        V^{\vartheta}(\cdot;P) = V^{\vartheta \odot N^{-1}}(\cdot;S).
    \end{equation}
    In particular, the map $\vartheta \mapsto \vartheta \odot N^{-1}$ is a bijection between $\mathcal{T}(P)$ and $\mathcal{T}(S)$, the sets of trading strategies with respect to $P$ and $S$, respectively.
\end{proposition}

By Propositions~\ref{prop: equivalence} and~\ref{prop: change of basis}, the relative wealth of a trading strategy $\vartheta$ with respect to $P$ in two bases $\mu$ and $\nu$ is given as
\[
    V^{\kappa\vartheta}(\cdot;\nu) = \frac{V^\vartheta(\cdot;P)}{\widetilde\Sigma(\cdot)},
    \qquad
    V^{\vartheta\odot N^{-1}}(\cdot;\mu) = \frac{V^\vartheta(\cdot;P)}{\Lambda(\cdot)},
\]
and dividing the two identities, provided that $V^{\vartheta}(\cdot;P)>0$, yields
\begin{equation} \label{eq: cross basis identity}
    \frac{V^{\vartheta\odot N^{-1}}(\cdot;\mu)}{V^{\kappa\vartheta}(\cdot;\nu)} = \frac{\widetilde\Sigma(\cdot)}{\Lambda(\cdot)}.
\end{equation}
The right-hand side is independent of $\vartheta$, so the change of basis converts the relative wealth of the same investment between two benchmarks. Taking logarithms and normalizing at $t=0$, the cross-basis log-wealth difference can be quantified by
\begin{equation} \label{eq: R definition}
    R(t) := \log \frac{\widetilde\Sigma(t)/\widetilde\Sigma(0)}{\Lambda(t)/\Lambda(0)}, \qquad t \ge 0,
\end{equation}
the cumulative log performance of the divisor-adjusted price-weighted basket relative to the capital\-ization-weighted one. For a fixed constituent universe, this identity provides the benchmark conversion used in Section~\ref{sec: numerics}. This correspondence concerns the same absolute wealth strategy expressed in two bases. It does not identify the strategies generated separately by evaluating a given function at $\mu$ and at $\nu$, which generally have different allocations.

\bigskip

\section{Empirical results} \label{sec: numerics}

We examine the price-based framework using daily NYSE data from 2002 to 2021. We construct both price-weighted and capitalization-weighted indices on each of two quarterly selected universes: the 500 stocks with the highest nominal prices and the 500 stocks with the largest market capitalizations. Constituents are selected using information available at each re-selection date, rather than fixing the universe retrospectively using the constituents at the end of the sample. In each universe, we use both the market weight vector $\mu$ of \eqref{def:mu} and the price weight vector $\nu$ of \eqref{def:nu} for generating trading strategies, to distinguish the weighting rule from the universe-selection rule.

This section proceeds in three parts: Section~\ref{sec: data} describes the data and universe construction, compares the cross-sectional distributions of capitalization and price weights, and presents empirical split divisors; Section~\ref{sec: vw vs pw} compares the price- and capitalization-weighted benchmarks, decomposes their performance gap into weighting and selection effects, and examines their cross-sectional diversity; and Section~\ref{sec: numerical fgp} examines the functionally generated diversity-weighted portfolios of Proposition~\ref{prop: multiplicatively generated trading strategy theorem}, with emphasis on the role of the jump correction terms induced by adjustment events.

\medskip

\subsection{Data and universes} \label{sec: data}

We use daily prices and outstanding shares from the Center for Research in Security Prices~(CRSP) for common stocks listed on the New York Stock Exchange~(NYSE) over 2002-01-02 to 2021-12-31, comprising $T+1=5{,}036$ trading-day observations, indexed by $t=0,\dots,T$. We write $P_i(t)$ for the closing price of stock $i$, $N_i(t)$ for its outstanding shares, and $S_i(t) = P_i(t)\,N_i(t)$ for its capitalization at day $t$. The capitalization weights $\mu$ and price weights $\nu$ are computed as in \eqref{def:mu} and \eqref{def:nu}, respectively.

The capital distribution curve of \cite[Figure~5.1]{fernholz2002stochastic} plots capitalization weights $\mu_i$ in decreasing order against their ranks on logarithmic axes; the price distribution curve is defined analogously, using $\nu_i$ in place of $\mu_i$. Figure~\ref{fig: rank weight full} shows both families for the full NYSE sample at six dates between 2002 and 2021, with approximately $2{,}400$--$2{,}700$ stocks at each date. Compared with the capital distribution curves, the price distribution curves exhibit somewhat greater variation across the selected years and are closer to linear among the largest weights, with less pronounced concavity on the log-log scale.

\begin{figure}[!htbp]
    \centering
    \includegraphics[width=0.98\textwidth]{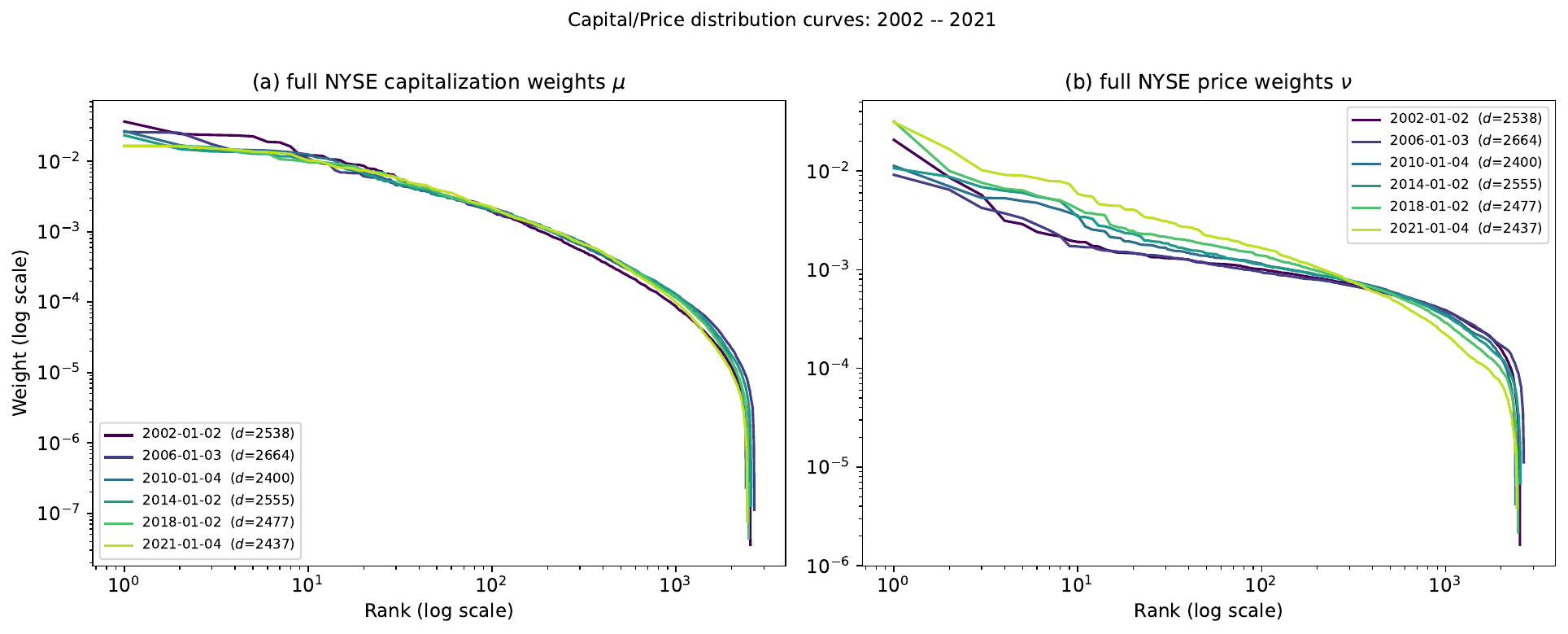}
    \caption{Capital (left) and price (right) distribution curves for the NYSE sample: 2002--2021.}
    \label{fig: rank weight full}
\end{figure}

For the benchmark and portfolio comparisons, we select two universes quarterly, one by nominal share price and the other by market capitalization. Let $\hat t_m$ be the first trading day of each quarter and $s_m$ the preceding trading day. The incoming universes are selected at the close of $s_m$ as
\[
    U^P_m := \big\{\text{top-}500 \text{ ranked according to } P_i(s_m)\big\},
    \quad
    U^S_m := \big\{\text{top-}500 \text{ ranked according to } S_i(s_m)\big\}.
\]
The return ending at $s_m$ uses the outgoing universe, and the return from the close of $s_m$ to the close of $\hat t_m$ uses the incoming universe. We select the initial universes at the first sample close, when all index levels are normalized to one.

Each universe is held fixed within its quarter and renewed at the beginning of the next quarter, giving an initial selection and $79$ rotations over $20$ years. The two universes overlap substantially but never coincide, sharing on average $268$ of $500$ names per quarter. An average of $51$ constituents are replaced per quarter in the price-selected universe, compared with $23$ in the capitalization-selected universe. Each universe is thus an \emph{open market} in the sense of \cite{Karatzas:Kim2} with quarterly entry and exit, rather than a \emph{closed market} with constant membership.

The empirical universes therefore differ from the fixed constituent market of Sections~\ref{sec: setup}--\ref{sec: FGPTS}. At each reconstitution, chain-linking preserves the outgoing wealth level while the incoming basket is initialized on the new universe. Constituent replacement is thus treated as a wealth-preserving switch, separately from the share-unit adjustment handled by the split divisor.

To visualize the distributional difference underlying these weights, Figure~\ref{fig: scatter} shows the cross-section $(\mu_i, \nu_i)$ on a log-log scale at the terminal date. Both weight vectors are computed on the stocks observed at that date in the union universe $U^P_m \cup U^S_m$ and normalized over this common sample. Points off the diagonal $\mu_i = \nu_i$ reflect heterogeneity in outstanding shares $N_i$: stocks above the line have fewer outstanding shares than the price-weighted average share count $\sum_j\nu_jN_j$ in the plotted sample, while those below have more. The dispersion spans several orders of magnitude in both directions, showing that the capitalization and price rankings need not coincide. This difference motivates analyzing the two weighting schemes separately.

\begin{figure}[!htbp]
    \centering
    \includegraphics[width=0.50\textwidth]{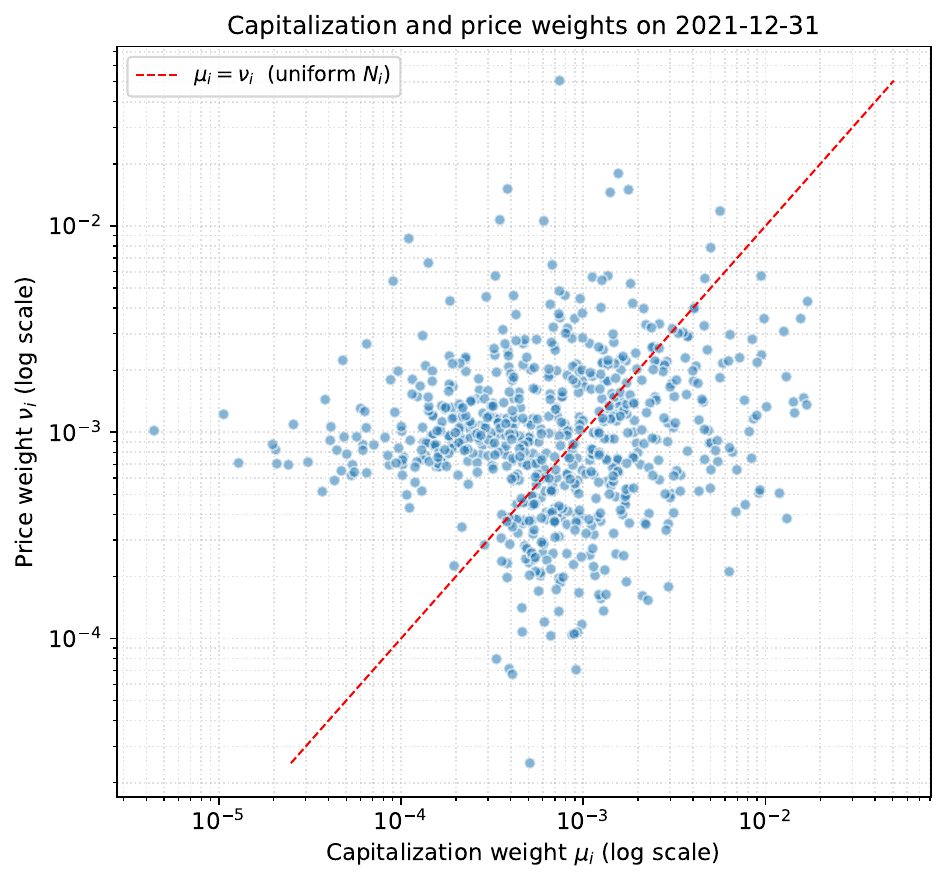}
    \caption{Capitalization and price weights $(\mu_i, \nu_i)$ for stocks observed at the terminal date in $U^P_m \cup U^S_m$, on a log-log scale. The dashed line indicates $\mu_i = \nu_i$.}
    \label{fig: scatter}
\end{figure}

Adjustment events are first inferred from joint changes in prices and outstanding shares. For each stock $i$ and day $t$, we compute the outstanding shares ratio $\rho_i(t) := N_i(t)/N_i(t-1)$ and the capitalization ratio $\eta_i(t) := S_i(t)/S_i(t-1)$. An $\ell$-for-$1$ split scales $N_i$ by $\ell$ and $P_i$ by $1/\ell$ at the event. To screen for this offsetting change in the daily observations, we initially assign
\[
    m_i(t) =
    \begin{cases}
        ~~1 & \text{if } |\rho_i(t)-1| \le \tau_{\mathrm{SH}}, \\[2pt]
        ~~1 & \text{if } |\rho_i(t)-1| > \tau_{\mathrm{SH}} \text{ and } |\eta_i(t)-1| \ge \tau_{\mathrm{MC}}, \\[2pt]
        \rho_i(t) & \text{if } |\rho_i(t)-1| > \tau_{\mathrm{SH}} \text{ and } |\eta_i(t)-1| < \tau_{\mathrm{MC}},
    \end{cases}
\]
with thresholds $\tau_{\mathrm{SH}}=\tau_{\mathrm{MC}}=0.15$. These are practical screening parameters: the share count threshold excludes small changes in outstanding shares while admitting split ratios such as $2$-for-$1$, $3$-for-$1$, $3$-for-$2$, $4$-for-$3$, and $5$-for-$4$. Among observations exceeding the share count threshold, we identify an adjustment event only when the change in capitalization remains below the capitalization threshold. This criterion is intended to capture a substantial change in outstanding shares that is approximately offset by an opposite change in the share price, while allowing for market price movement on the event day. Observations with substantial changes in both outstanding shares and capitalization may instead reflect share issuance or share cancellations and are not classified as adjustment events by this screen. Accordingly, $m_i(t)=1$ in these cases indicates that no split-type adjustment is applied, not that no corporate action occurred. Identified adjustments affecting different stocks on the same date are grouped into a single multiplier vector $m_k$ in the adjustment sequence $(\tau_k,m_k)$.

Daily changes in observed market capitalization may reflect updates to reported outstanding shares as well as market price movements. Stock splits additionally change the units in which prices are quoted. We therefore introduce \emph{adjusted starting values} for both prices and capitalizations. Fix $M \in \{P, S\}$ and a quarter $m$ for the universe $U^M_m$. For each constituent $i \in U_m^M$ and each trading day $\hat{t}_m \le t < \hat{t}_{m+1}$ with $t \ge 1$, define
\[
    P_i^\circ(t):=\frac{P_i(t-1)}{m_i(t)}, \qquad S_i^\circ(t):=P_i^\circ(t)N_i(t).
\]
On a reconstitution date $t = \hat{t}_{m} \ge 1$, the preceding closing prices are those of the incoming constituents $i \in U_m^M$, evaluated at $s_m$.

For a pure split, \(N_i(t)=m_i(t)N_i(t-1)\), so \(S_i^\circ(t)=S_i(t-1)\), preserving the split invariance of capitalization. We use the same share count in \(S_i^\circ(t)\) and \(S_i(t)\), so the measured return reflects price changes alone. Normalizing the starting prices \(P^\circ(t)\) and capitalizations \(S^\circ(t)\) gives the starting weights \(\nu^\circ(t)\) and \(\mu^\circ(t)\), respectively, which we use in the daily portfolio calculations below.

For the price-weighted benchmark on each universe $U_m^M$, $M\in\{P,S\}$, we construct a split divisor $\kappa^M$ by applying the day's split factors to the preceding closing prices:
\begin{equation} \label{eq: empirical divisor}
    \kappa^M(t)=\kappa^M(t-1)\frac{\sum_{i\in U_m^M}P_i^\circ(t)}{\sum_{i\in U_m^M}P_i(t-1)}, \qquad \kappa^M(0)=1.
\end{equation}
Both sums use the universe active on day $t$. The divisor is initialized at one on the first observation date and remains constant on days without a split adjustment to a current constituent.

Figure~\ref{fig: divisor} displays $\kappa^P$ for the price-selected universe $U^P_m$. Of the $765$ detected events, $753$ events were forward splits and the remaining $12$ were reverse splits, according to whether the adjustment factor is greater or less than one. The predominance of forward splits is reflected in the almost monotonic decrease of $\kappa^P$, from $1$ to about $0.41$ at the end of the sample. The decline is steep until 2008 and flattens afterwards, alongside a decrease in detected adjustment events among constituents of $U^P_m$.

\begin{figure}[!htbp]
    \centering
    \includegraphics[width=0.78\textwidth]{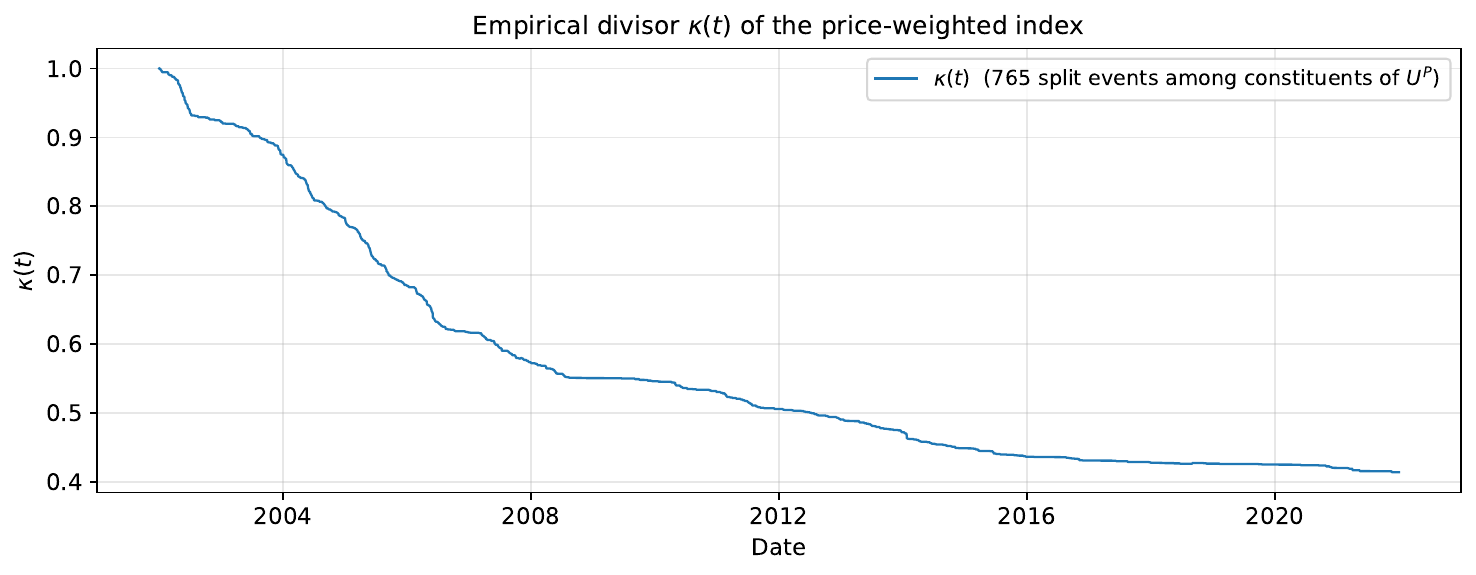}
    \caption{Empirical split divisor $\kappa^P(t)$ for the quarterly price-selected universe $U_m^P$, computed using \eqref{eq: empirical divisor}. The divisor changes only at split adjustments to current constituents.}
    \label{fig: divisor}
\end{figure}

\medskip

\subsection{Reference benchmarks and cross-sectional diversity} \label{sec: vw vs pw}

To compare the two weighting schemes ($\mu$ versus $\nu$) on each universe, we first construct price- and capitalization-weighted reference indices. For each $M\in\{P,S\}$, initialize $I_\mu^M(0)=I_\nu^M(0)=1$ and, for $t=1,\ldots,T$, define the chain-linked reference indices by
\[
    I_\mu^M(t) := I_\mu^M(t-1) \frac{\sum_{i\in U_m^M}S_i(t)}{\sum_{i\in U_m^M}S_i^\circ(t)}, \qquad
    I_\nu^M(t) := I_\nu^M(t-1) \frac{\sum_{i\in U_m^M}P_i(t)}{\sum_{i\in U_m^M}P_i^\circ(t)},
\]
where $m$ is determined by $\hat t_m\le t<\hat t_{m+1}$. We next define
\[
    L^M_w(t):=\log\frac{I^M_w(t)}{I^M_w(0)}, \qquad M\in\{P,S\}, \quad w\in\{\mu,\nu\},
\]
as the cumulative log level. Split divisors remove changes in the price sum caused solely by share-unit adjustments, whereas quarterly chain-linking separately removes level changes caused by replacing constituents. Neither correction removes the market price movement over the trading day. The price-weighted indices $I^P_\nu$ and $I^S_\nu$ use the divisor of Definition~\ref{def: divisor} on their respective universes, whereas the capitalization-weighted indices need none. Figure~\ref{fig: vw pw} compares $L^M_w$ with the two weighting schemes on each universe. Over 2002--2021, the capitalization-weighted indices end at $L^S_\mu(T)=0.97$ and $L^P_\mu(T)=1.10$, and the price-weighted indices at $L^S_\nu(T)=1.45$ and $L^P_\nu(T)=1.61$, so that price weighting ends higher on both universes; all four trace the same crisis drawdowns in 2008 and 2020.

\begin{figure}[!htbp]
    \centering
    \includegraphics[width=0.98\textwidth]{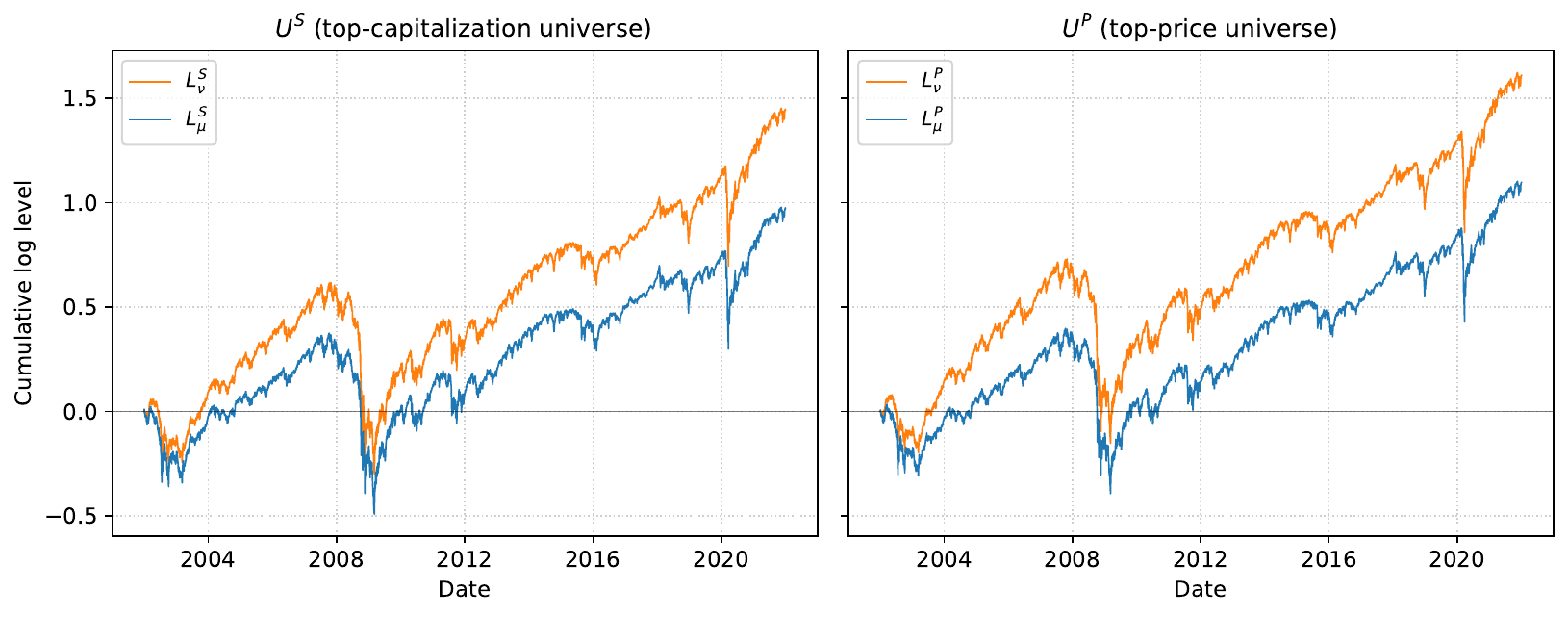}
    \caption{Cumulative log levels $L^M_w$ of the chain-linked reference indices, compared on the same universe: $M = S$ (left) and $M = P$ (right), each with $w = \mu$ (blue) and $w = \nu$ (orange).}
    \label{fig: vw pw}
\end{figure}

The headline comparison of the two frameworks is the gap between the log level $L^P_\nu$ of the price-weighted index on the price universe and the log level $L^S_\mu$ of the capitalization-weighted index on the capitalization universe:
\begin{equation} \label{eq: headline gap}
    D(t) := L^P_\nu(t) - L^S_\mu(t), \qquad 0 \le t \le T,
\end{equation}
which reaches $D(T)=0.63$. The two indices in \eqref{eq: headline gap} differ in two respects at once: the weighting scheme and the universe, so $D$ mixes the effect of price weighting with the effect of selecting high-priced stocks. The four indices allow us to separate the two effects.
Write
\begin{equation}    \label{def:W-C}
    W^P := L^P_\nu - L^P_\mu, \qquad W^S := L^S_\nu - L^S_\mu, \qquad
    C^\mu := L^P_\mu - L^S_\mu, \qquad C^\nu := L^P_\nu - L^S_\nu.
\end{equation}
Here $W^P$ and $W^S$ are the weighting effects, the gain from replacing capitalization weights by price weights on a fixed universe, and $C^\mu$ and $C^\nu$ are the selection effects, the gain from replacing the capitalization universe by the price universe under a fixed weighting scheme. On a fixed universe, the weighting effect is the empirical counterpart of the cross-basis quantity $R$ in \eqref{eq: R definition}, with universe rotation absorbed by chain-linking. The difference $D$ admits two exact decompositions,
\begin{equation} \label{eq: gap decomposition}
    D = W^P + C^\mu = W^S + C^\nu,
\end{equation}
depending on whether the weighting scheme is changed first and the universe second, or the other way around; the two paths differ by the interaction $W^P - W^S = C^\nu - C^\mu$.

\begin{figure}[!htbp]
    \centering
    \includegraphics[width=0.98\textwidth]{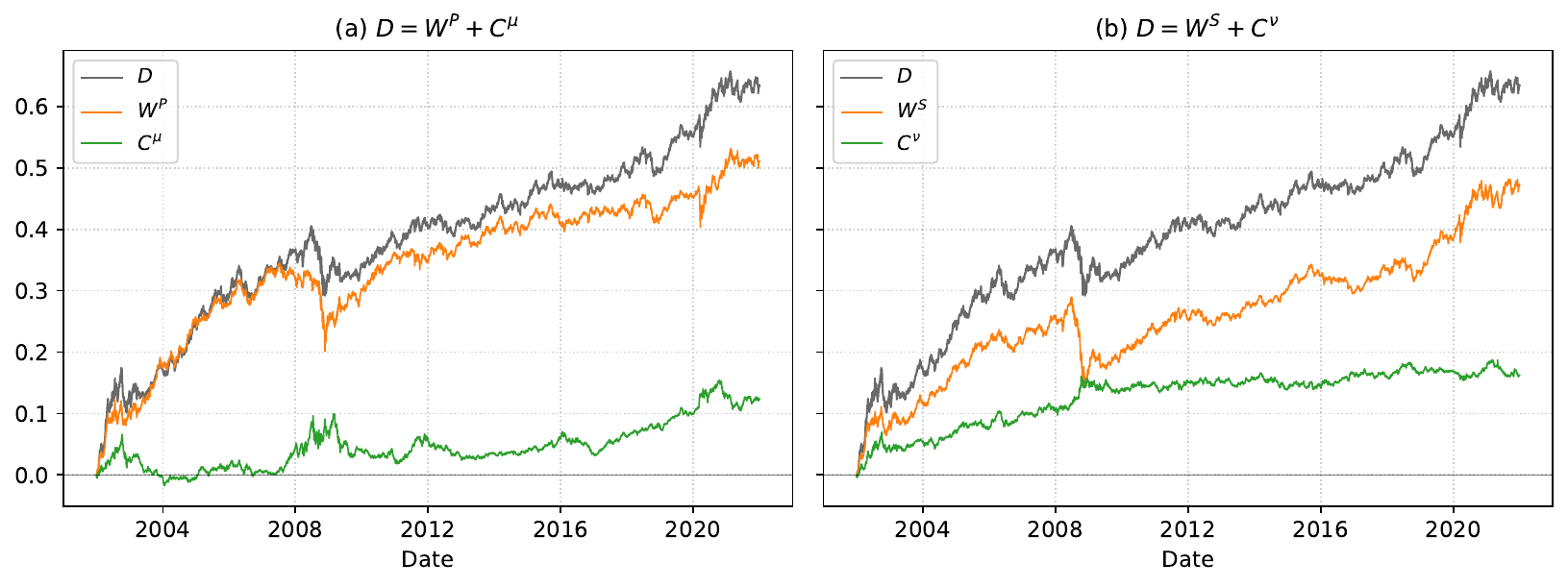}
    \caption{The two decompositions \eqref{eq: gap decomposition} of the gap $D=L^P_\nu-L^S_\mu$ (grey): (a) $D=W^P+C^\mu$ and (b) $D=W^S+C^\nu$.}
    \label{fig: decomposition}
\end{figure}

Figure~\ref{fig: decomposition} shows that the within-universe weighting differences $W^P$ and $W^S$ account for most of $D(T)$, while the universe-selection differences $C^\mu$ and $C^\nu$ make smaller contributions. Moreover, $W^P$ and $W^S$ are similar, so the higher terminal performance of price weighting is observed under both universe-selection rules rather than only among the stocks selected by nominal price. A possible explanation for the larger weighting effects is that changing the weighting scheme alters portfolio diversification more substantially than changing the constituent universe. We therefore examine how the diversity of the benchmark weights varies with the weighting scheme and the constituent universe.

The Gibbs entropy function \eqref{eq:Gibbs-entropy} is a standard measure of diversity in stochastic portfolio theory \cite{fernholz2002stochastic}. For each $M \in \{P,S\}$, write $\mu^M$ and $\nu^M$ for the capitalization weights and the price weights, respectively, computed on $U^M_m$. Figure~\ref{fig: entropy} shows $H(\mu^M(t))$ and $H(\nu^M(t))$ for each universe, together with the gap $H(\nu^M(t))-H(\mu^M(t))$. On the price-selected universe $U_m^P$, price weights have higher entropy than capitalization weights on every observation date. On the capitalization-selected universe $U_m^S$, the same ordering holds for most of the sample but reverses in 2020.

\begin{figure}[!htb]
    \centering
    \includegraphics[width=0.98\textwidth]{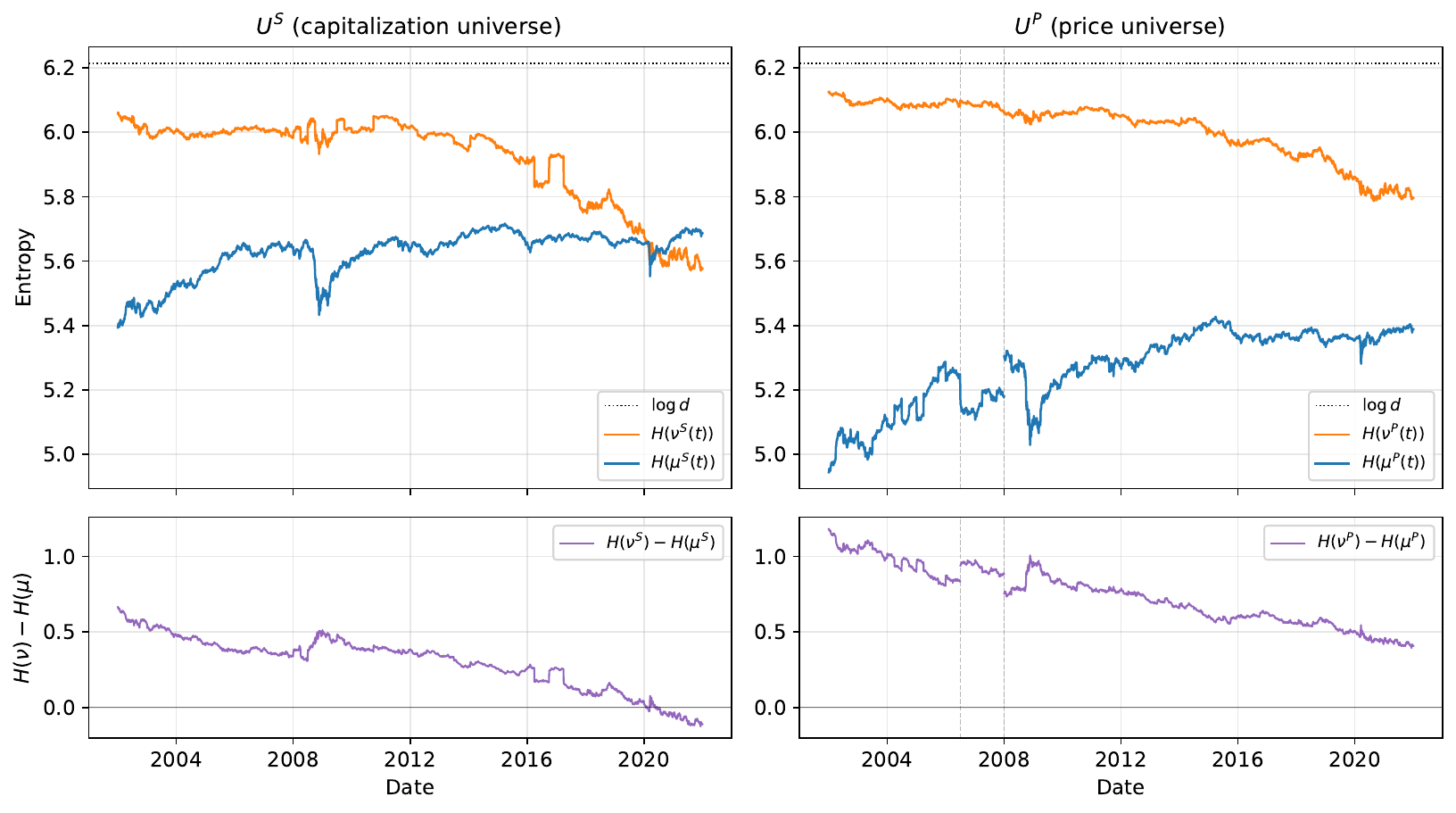}
    \caption{Top: daily Gibbs entropy of the capitalization weights (blue) and the price weights (orange) on the capitalization universe $U^S_m$ (left) and the price universe $U^P_m$ (right). The dotted line marks the maximum $\log d = \log 500 \approx 6.21$. Bottom: the gap $H(\nu^M(t)) - H(\mu^M(t))$ on each universe. The vertical dashed lines in the right-hand panels mark quarterly reconstitutions associated with pronounced changes in $H(\mu^P)$.}
    \label{fig: entropy}
\end{figure}

This observation is consistent with \cite{damjanovic2025price}, which found in simulations on common baskets of S\&P~500 constituents that price weighting tends to yield more diversified portfolios than capitalization weighting. Our results show that the entropy difference depends on the selected universe. On the price-selected universe, capitalization weights are substantially more concentrated than price weights. This contrast is less pronounced on the capitalization-selected universe, where capitalization weights have higher entropy than on the price-selected universe. Near the end of the sample, the ordering reverses on the capitalization-selected universe $U_m^S$, where price weights have lower entropy than capitalization weights.

The decline in price-weighted entropy may partly reflect less frequent splitting of already high-priced stocks, allowing their prices to rise further relative to those of other constituents.
We also notice that the price-weighted entropy is more robust than the capitalization-weighted one on both universes; during the 2008 and 2020 crises, the capitalization-weighted entropy dropped significantly in a short amount of time, whereas the price-weighted entropy did not change much.

\FloatBarrier

\medskip

\subsection{Functionally generated diversity-weighted portfolios} \label{sec: numerical fgp}

This subsection examines the empirical performance of the diversity-weighted portfolio in Example~\ref{ex: diversity weighted}, multiplicatively generated by the diversity-weighted family $\mathbf{D}_p$ for $0<p<1$. The $p$-diversity-weighted portfolio $\pi^{(p)}_i(w) = w_i^p / \sum_j w_j^p$ is evaluated at a weight vector $w \in \{\mu, \nu\}$ satisfying $\mathbf{1} \cdot w \equiv 1$; as $p \uparrow 1$, it recovers the underlying weights $w$, and decreasing $p$ to zero progressively flattens the allocation toward equal weighting. We report $p \in \{0.1, 0.3, 0.5, 0.7, 0.9\}$, evaluated under the capitalization weights $\mu^M$ and under the price weights $\nu^M$ on each of the two universes $M \in \{P, S\}$.

Figure~\ref{fig: fgp} illustrates the log relative value $\log V^{\psi}$ in \eqref{eq: value process of multiplicatively generated trading strategy} for $G = \mathbf{D}_p$ with the five values of $p$, on $U^S_m$ (left column) and on $U^P_m$ (right column). Panels (a) and (b) use the capitalization weights $\mu^M$ as the generating basis, so the relative value is measured against the capitalization-weighted index $I^M_\mu$. Panels (c) and (d) use the price weights $\nu^M$ and are measured against the price-weighted index $I^M_\nu$. Panels (e) and (f) add the weighting effects $W^S$ and $W^P$ of \eqref{def:W-C} respectively, to the curves of (c) and (d), so that the price-based strategies are also measured against $I^M_\mu$; this is the cross-basis identity \eqref{eq: cross basis identity}, with $W^M$ playing the role of $R$. All reported log relative wealth curves are normalized to zero at the initial observation.

At the terminal date, the capitalization-based portfolios in panels (a) and (b) have positive log relative wealth, with larger gains for smaller $p$. The price-based portfolios in panels (c) and (d) accumulate substantial relative gains during 2008--2009, but much of these gains is subsequently eroded. Their terminal log relative wealth is close to zero on $U_m^S$ and remains modestly positive on $U_m^P$.

After rebasing to the capitalization benchmark, however, the price-based portfolios in panels (e) and (f) finish above the corresponding capitalization-based portfolios for every reported $p$. Most of their terminal log relative wealth against this benchmark is accounted for by $W^M$, rather than by their gain over the price-weighted benchmark. Adding $W^M$ leaves the numerical differences between the price-based curves for different $p$ unchanged.

\begin{figure}[!t]
    \centering
    \includegraphics[width=0.98\textwidth]{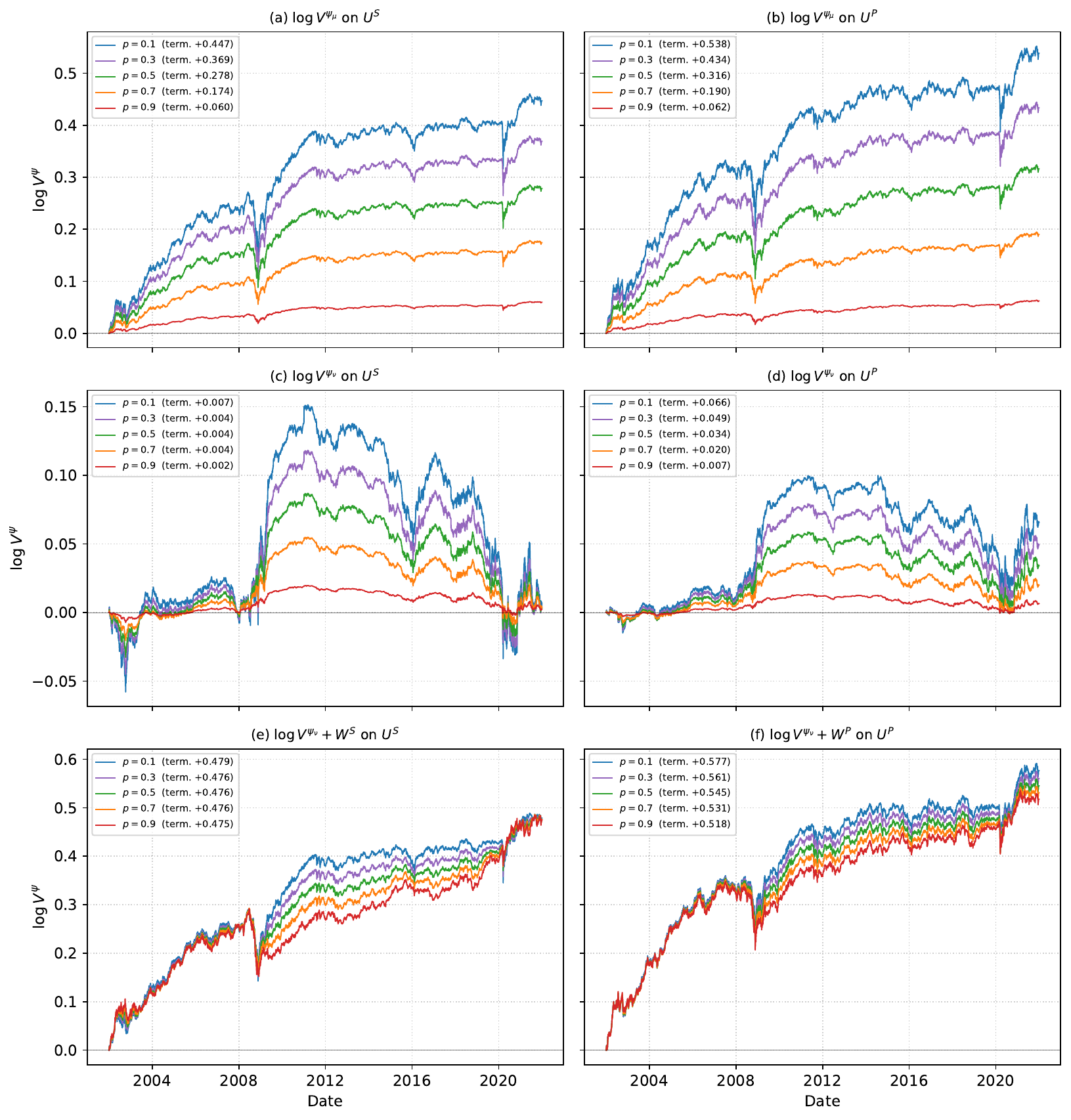}
    \caption{Log relative wealth of the diversity-weighted portfolios generated by $\mathbf D_p$, for $p\in\{0.1,0.3,0.5,0.7,0.9\}$. The left and right columns use $U_m^S$ and $U_m^P$, respectively. Panels (a) and (b) use capitalization weights and the corresponding benchmark $I_\mu^M$; panels (c) and (d) use price weights and the corresponding benchmark $I_\nu^M$. Panels (e) and (f) show the same price-based strategies as (c) and (d), rebased to $I_\mu^M$ by adding $W^M$. Vertical scales are shared within each row but differ across rows.}
    \label{fig: fgp}
\end{figure}

To better compare the relative performance of these portfolios under different weighting schemes and universes, we decompose the normalized log relative wealth via the master equation \eqref{eq: value process of multiplicatively generated trading strategy}:
\begin{equation} \label{eq: fgp decomposition empirical}
    \log\frac{V^\psi(t)}{V^\psi(0)} = \underbrace{(\log \mathbf{D}_p)^c\big(w(t)\big)-(\log \mathbf{D}_p)^c\big(w(0)\big)}_{\text{continuous state}} + \underbrace{\int_0^t\frac{d\Gamma^{\mathbf{D}_p,c}(s)}{\mathbf{D}_p\big(w(s-)\big)}}_{\text{excess growth}}.
\end{equation}

For the formulas below, fix $M\in\{P,S\}$ and suppress the universe superscript. The choice $w=\mu$ or $w=\nu$ specifies the generating basis, and all sums over constituents use the universe active on the relevant day.

For daily stock data, we denote by $C_p(t)$ and $E_p(t)$ the following approximations to the continuous state and excess growth terms in \eqref{eq: fgp decomposition empirical}, respectively. Let $w(t)$ denote the closing weight vector (either $\mu$ or $\nu$) on trading day $t$, and let $w^\circ(t)$ denote the starting weight vector constructed from the preceding closing prices after applying the adjustments effective on day $t$. At a universe reconstitution, these starting weights are formed on the incoming universe. In the daily approximation, $w^\circ(t)$ is the left endpoint for the price movement retained in that day's increment. On a split day, it is a post-adjustment starting state, not the pre-adjustment left limit of the generating state at the split time.

Explicitly, for price weights $\nu$, the starting weights are
\[
    \nu_i^\circ(s) := \frac{P_i^\circ(s)}{\sum_j P_j^\circ(s)}.
\]
On a day with neither a split nor a universe reconstitution, $\nu^\circ(s)=\nu(s-1)$. For capitalization weights $\mu$, the starting weights are computed as
\[
    \mu_i^\circ(s) := \frac{S_i^\circ(s)}{\sum_j S_j^\circ(s)}.
\]
A pure split leaves market capitalization unchanged: since $N_i(s)=m_i(s)N_i(s-1)$, the numerator reduces to $S_i(s-1)$. Thus, in the absence of other share count changes and a universe reconstitution, $\mu^\circ(s)=\mu(s-1)$.

A split day $s$ without a universe reconstitution therefore involves the two successive changes
\[
    \nu(s-1)\longrightarrow\nu^\circ(s)\longrightarrow\nu(s).    
\]
The first change is due to the split adjustment, whereas the second reflects the subsequent price movement. Accordingly, the log change in the diversity function can be decomposed as
\begin{equation} \label{eq: daily state change decomposition}
    \log\frac{\mathbf D_p\big(\nu(s)\big)}{\mathbf{D}_p\big(\nu(s-1)\big)} = \log\frac{\mathbf{D}_p\big(\nu^\circ(s)\big)}{\mathbf{D}_p\big(\nu(s-1)\big)} + \log\frac{\mathbf{D}_p\big(\nu(s)\big)}{\mathbf{D}_p\big(\nu^\circ(s)\big)}.
\end{equation}
The first term on the right-hand side is the split-induced jump in the generating state; hence it is not contained in the continuous state term of the wealth decomposition \eqref{eq: fgp decomposition empirical}, whereas the second term is retained in $C_p(t)$. We therefore approximate the continuous-time state term by the cumulative adjusted state component
\begin{equation} \label{eq: daily state component}
    C_p(t):=\sum_{s=1}^t\log\frac{\mathbf D_p\big(w(s)\big)}{\mathbf D_p\big(w^\circ(s)\big)}.
\end{equation}
The excess growth component is approximated by
\begin{equation} \label{eq: daily excess growth component}
    E_p(t):=-\frac12\sum_{s=1}^t\sum_{i,j} \frac{\partial_{ij}^2\mathbf D_p\big(w^\circ(s)\big)}{\mathbf D_p\big(w^\circ(s)\big)} \big(w_i(s)-w_i^\circ(s)\big)\big(w_j(s)-w_j^\circ(s)\big).
\end{equation}
Accordingly, $C_p(t)+E_p(t)$ approximates the normalized log relative wealth in \eqref{eq: fgp decomposition empirical}.

Figures~\ref{fig: fgp state} and \ref{fig: fgp excess growth} show the adjusted state $C_p(t)$ and excess growth $E_p(t)$ components of the decomposition \eqref{eq: fgp decomposition empirical} computed via \eqref{eq: daily state component} and \eqref{eq: daily excess growth component}, respectively. In both figures, the left and right columns correspond to the universes $U^S_m$ and $U^P_m$, respectively, while the upper and lower rows use capitalization weights $\mu^M$ and price weights $\nu^M$ as input, respectively.

The excess growth estimates in Figure~\ref{fig: fgp excess growth} are of comparable magnitude across the two generating bases and are larger for smaller $p$. The adjusted state components in Figure~\ref{fig: fgp state}, however, differ substantially. For price weights, their sustained decline offsets most of the excess growth accumulation, leaving only modest own-benchmark gains. For capitalization weights, the state decline is smaller, so a larger fraction of the excess growth accumulation remains in relative wealth. Thus the difference in own-benchmark performance is associated primarily with the adjusted state components, rather than with substantially weaker excess growth accumulation under price weighting.

These comparisons concern different portfolios measured against different benchmarks. The common-benchmark comparison in panels (e) and (f) of Figure~\ref{fig: fgp} separates the benchmark conversion from the price-based portfolios' own-benchmark gains.

\begin{figure}[!htb]
    \centering
    \includegraphics[width=\textwidth]{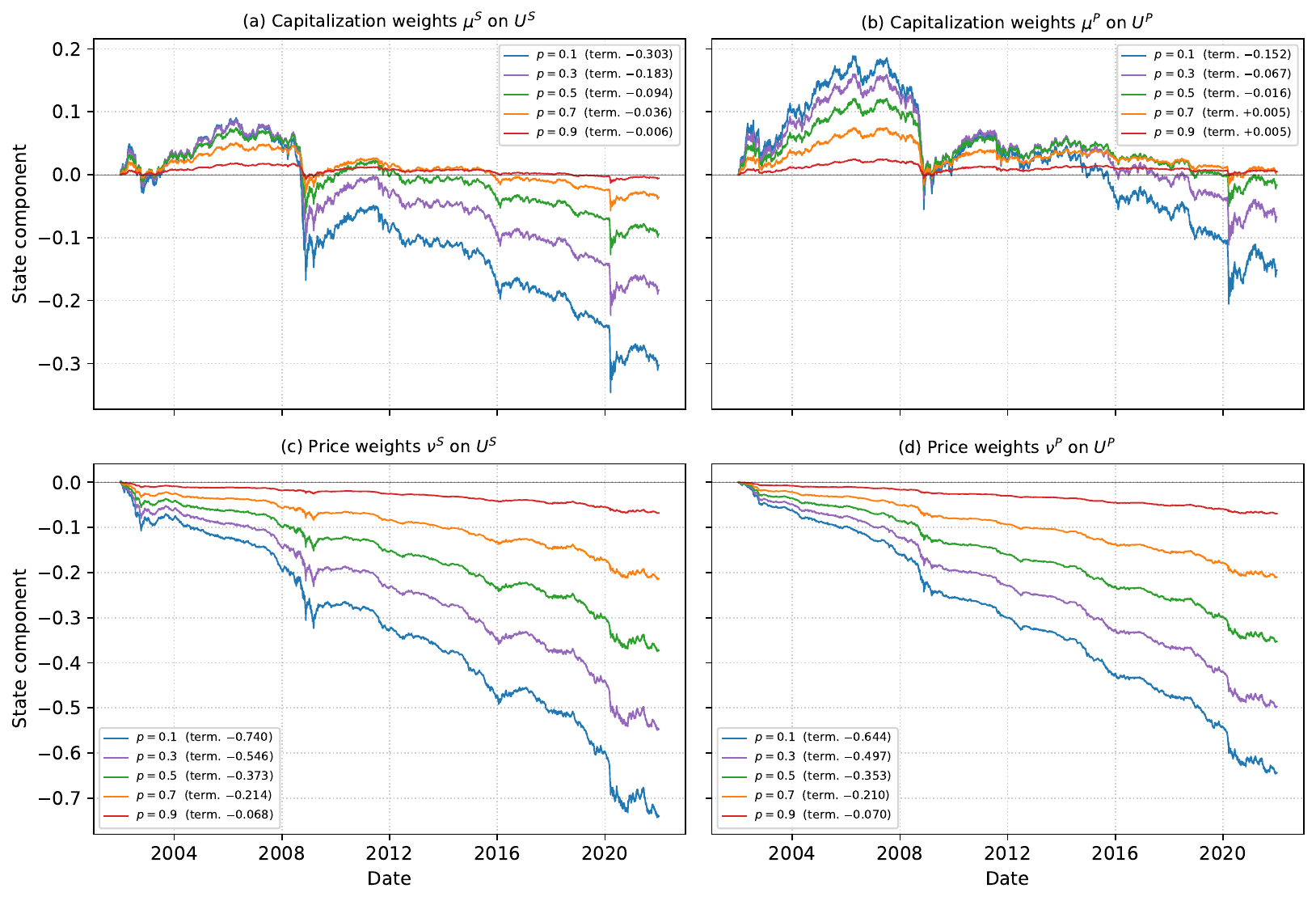}
    \caption{Adjusted state component $C_p(t)$ of \eqref{eq: daily state component}. The left and right columns use the universes $U^S_m$ and $U^P_m$, respectively. The first and second rows use capitalization weights $\mu^M$ and price weights $\nu^M$, respectively. The vertical scale is shared within each row.}
    \label{fig: fgp state}
\end{figure}

\begin{figure}[!htb]
    \centering
    \includegraphics[width=\textwidth]{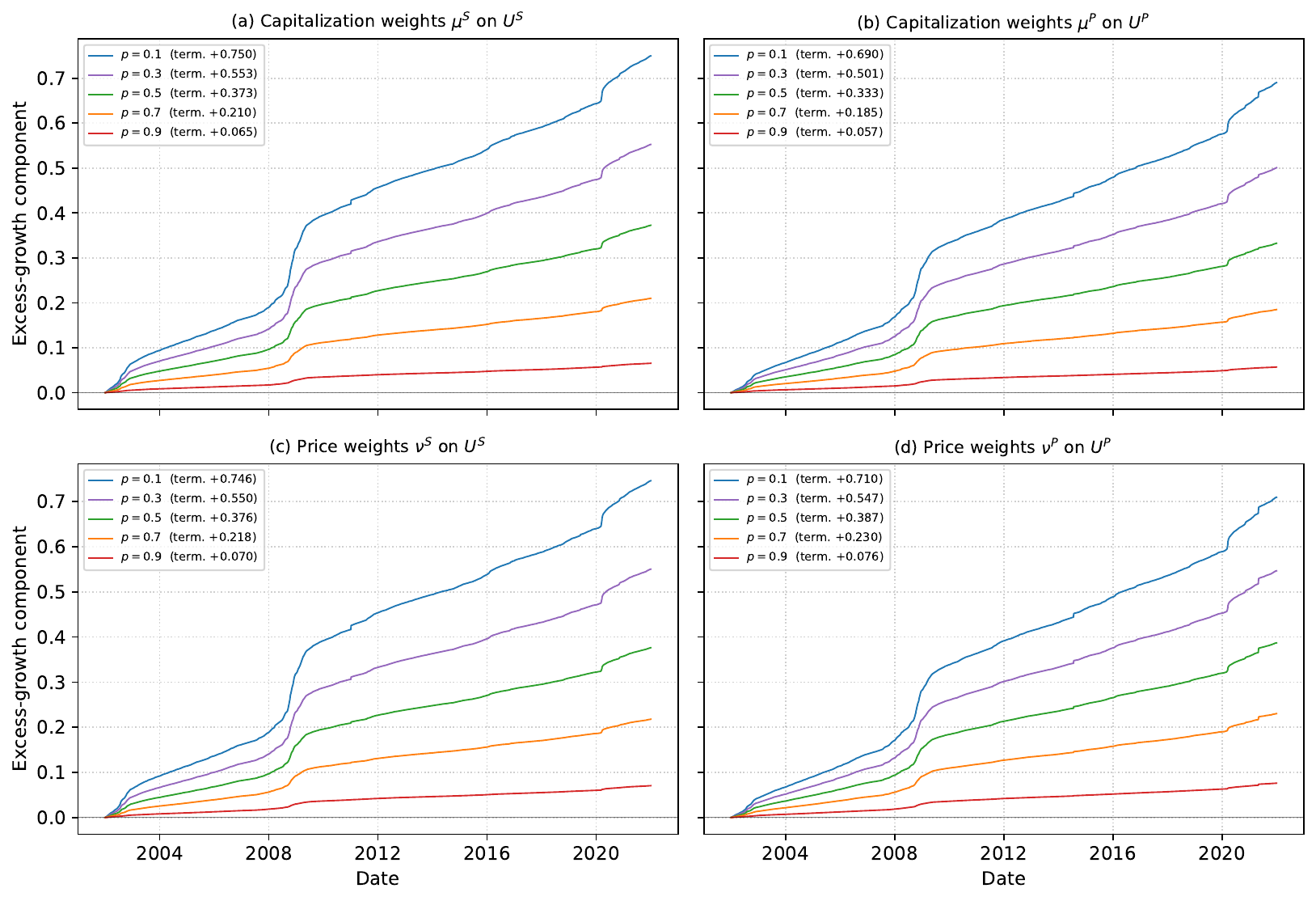}
    \caption{Excess growth component $E_p(t)$ of \eqref{eq: daily excess growth component}. The left and right columns use the universes $U^S_m$ and $U^P_m$, respectively. The first and second rows use capitalization weights $\mu^M$ and price weights $\nu^M$, respectively. The vertical scale is shared within each row.}
    \label{fig: fgp excess growth}
\end{figure}

We next quantify the split-induced changes in the generating state that are excluded from the adjusted state component. Recall that the first term on the right-hand side of \eqref{eq: daily state change decomposition} collects jumps induced by adjustment events. For the price weights $\nu$, we denote
\begin{equation} \label{eq: cumulative split jump}
    J_p(t) := \sum_{s\in\mathcal S, s\le t} \log\frac{\mathbf D_p\big(\nu^\circ(s)\big)}{\mathbf D_p\big(\nu(s-1)\big)},
\end{equation}
where $\mathcal{S}$ is the set of split days, excluding universe reconstitution dates. Each summand measures the change in $\log\mathbf{D}_p$ from the preceding closing weight to the post-split starting weight. Thus, $J_p$ isolates split effects on non-reconstitution days. The construction of $C_p$ also excludes changes arising from universe reconstitution.

\begin{figure}[htbp]
    \centering
    \includegraphics[width=\textwidth]{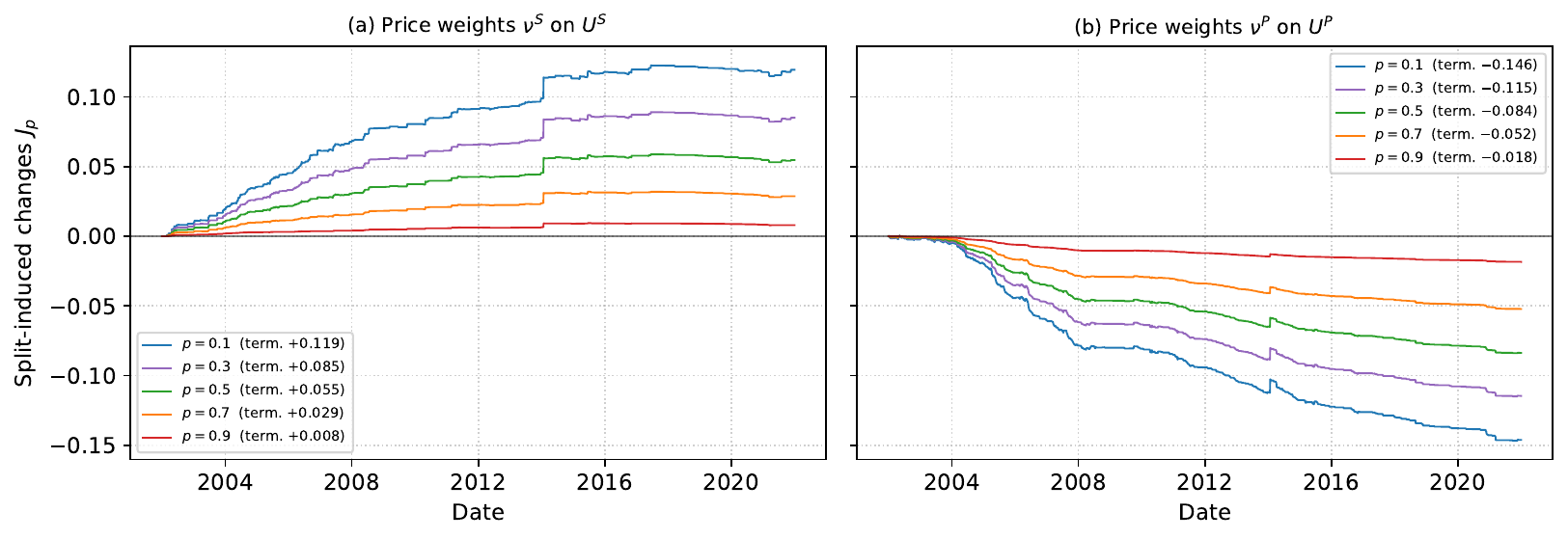}
    \caption{Cumulative split-induced jumps $J_p$ of \eqref{eq: cumulative split jump} for $p$-diversity-weighted portfolios with price weights on $U^S_m$ (left) and $U^P_m$ (right). These jumps are excluded from the adjusted state components in panels (c) and (d) of Figure~\ref{fig: fgp state}. Both panels share the same vertical scale.}
    \label{fig: fgp split jump}
\end{figure}

Figure~\ref{fig: fgp split jump} displays the cumulative split-induced jumps $J_p$ to illustrate the role of the split correction in the relative wealth of $p$-diversity-weighted portfolios. A striking feature is that $J_p$ shows an overall upward trend in the capitalization-selected universe $U^S_m$ but a downward trend in the other universe $U^P_m$. One possible explanation concerns the splitting stock's price relative to the other constituents. In $U^S_m$, selection by capitalization allows a broad range of share prices, and stocks undergoing forward splits tend to have relatively large price weights. Reducing these weights by forward splits can make the distribution more even, which helps increase diversity. In $U^P_m$, however, the other constituents are also selected because of their high prices, so a high absolute price need not imply a large relative weight among the 500 highest-priced stocks. If a splitting stock already has a price weight $\nu_i$ below the equal-weight level $1/d$ in $U^P_m$, then further reducing its price by forward splits can make the distribution less even, negatively impacting price diversity.

The signs of $J_p$ describe mechanical changes in the generating state, not instantaneous gains or losses of the self-financing portfolio. These changes are excluded from the wealth decomposition with the opposite sign. Holding all other components fixed, omitting the displayed split correction would therefore overstate log relative wealth on $U_m^S$ and understate it on $U_m^P$. This comparison concerns an omitted accounting correction, not a counterfactual market in which the splits did not occur.

All figures in this section are produced from a single reproducible Python pipeline over the NYSE data from CRSP.

\FloatBarrier

\bigskip

\appendix

\section{Proofs of Section~\ref{sec: setup}} \label{app: proofs setup}

This appendix collects the proofs of the results of Section~\ref{sec: setup}, in the order of their appearance.

\begin{proof}[Proof of Proposition~\ref{prop: equivalence}]
    The last statement follows from the definitions \eqref{def:nu}, \eqref{def:value-process} of $\nu$ and $V$. For the first assertion, suppose that $\vartheta$ is a trading strategy with respect to $P$. Then, since $V^\vartheta(\cdot; P)$ is continuous, we have
    \begin{equation} \label{eq: equiv jump identity}
        \vartheta(t-) \cdot \Delta P(t) = \vartheta(t-) \cdot \big( P(t) - P(t-) \big) = - \Delta \vartheta(t) \cdot P(t), \qquad t > 0.
    \end{equation}
    Now, apply It\^o's product rule to $V^{\kappa \vartheta}(\cdot; \nu) = V^\vartheta(\cdot; P) \Sigma^{-1}(\cdot) \kappa(\cdot)$:
    \begin{equation*}
        d\big(V^{\vartheta}(t; P) (\Sigma^{-1}\kappa)(t)\big) = 
        \underbrace{V^{\vartheta}(t-; P) d(\kappa \Sigma^{-1})(t)}_{\text{(A)}} + 
        \underbrace{(\kappa\Sigma^{-1})(t-) dV^{\vartheta}(t; P)}_{\text{(B)}} + 
        \underbrace{d[V^{\vartheta}(\cdot; P), \kappa \Sigma^{-1}](t)}_{\text{(C)}}.
    \end{equation*}
    We now compute each term separately. For (A),
    \begin{align*}
        V^{\vartheta}(t-; P) d(\kappa \Sigma^{-1})(t) &= \vartheta(t-) \cdot P(t-) \Big( \kappa(t-) d\Sigma^{-1}(t) + \Sigma^{-1}(t-) d\kappa(t) + d[\kappa, \Sigma^{-1}](t) \Big) \\
        &= \vartheta(t-) \cdot P(t-) \Big(\kappa(t-) d\Sigma^{-1}(t) + \Sigma^{-1}(t) d\kappa(t) \Big) \\
        &= \kappa(t-) \vartheta(t-) \cdot P(t-) d\Sigma^{-1}(t) + \vartheta(t) \cdot \nu(t) d\kappa(t).
    \end{align*}
    Here, the second equality uses
    \[
        d[\kappa, \Sigma^{-1}](t) = \Delta\Sigma^{-1}(t)\,d\kappa(t),
    \]
    which holds because $\kappa$ is constant between adjustment events. The last equality follows from the continuity of $V^{\vartheta}(\cdot;P)=\vartheta\cdot P$ and the identity $\nu=P\Sigma^{-1}$. For (B), define the cumulative jump process
    \[
        J^\vartheta(t) := \sum_{0<s\le t}\Delta\vartheta(s), \qquad t\ge0.
    \]
    Since $\vartheta$ has locally finitely many jumps, $J^\vartheta$ is of finite variation. The self-financing condition gives
    \begin{align*}
        (\kappa \Sigma^{-1})(t-) dV^\vartheta(t;P) &= (\kappa\Sigma^{-1})(t-) \Big(\vartheta(t-)\cdot dP(t) + P(t)\cdot dJ^\vartheta(t) \Big) \\
        &= (\kappa\Sigma^{-1})(t-) \vartheta(t-)\cdot dP(t) + \kappa(t)\nu(t)\cdot dJ^\vartheta(t).
    \end{align*}
    Here, the last equality follows from the continuity of $\kappa\Sigma^{-1}$ and the identity $\nu=P\Sigma^{-1}$. Finally, for (C), the self-financing condition and the continuity of $\kappa\Sigma^{-1}$ give
    \begin{align*}
        &d[V^\vartheta(\cdot;P),\kappa\Sigma^{-1}](t) \\
        &\quad= \kappa(t-)\vartheta(t-)\cdot d[P,\Sigma^{-1}](t) + \Sigma^{-1}(t-)\vartheta(t-)\cdot d[P,\kappa](t) + \vartheta(t-)\cdot d[P,[\kappa,\Sigma^{-1}]](t) \\
        &\quad= \kappa(t-)\vartheta(t-)\cdot d[P,\Sigma^{-1}](t) + \vartheta(t-)\cdot\Delta P(t) \Big(\Sigma^{-1}(t-)+\Delta\Sigma^{-1}(t)\Big) d\kappa(t) \\
        &\quad= \kappa(t-)\vartheta(t-)\cdot d[P,\Sigma^{-1}](t) + \Sigma^{-1}(t)\vartheta(t-)\cdot\Delta P(t) d\kappa(t) \\
        &\quad= \kappa(t-)\vartheta(t-)\cdot d[P,\Sigma^{-1}](t) - \Delta\vartheta(t)\cdot\nu(t)\,d\kappa(t).
    \end{align*}
    Here, the second equality uses
    \[
        d[P,\kappa](t)=\Delta P(t)\,d\kappa(t), \qquad d[P,[\kappa,\Sigma^{-1}]](t) =\Delta P(t)\Delta\Sigma^{-1}(t)\,d\kappa(t),
    \]
    which hold because $\kappa$ is constant between adjustment events. The last equality follows from~\eqref{eq: equiv jump identity} and $\nu=P\Sigma^{-1}$. Combining these terms and integrating, we obtain
    \[
        V^{\kappa\vartheta}(t;\nu) -V^{\kappa\vartheta}(0;\nu) = \int_0^t(\kappa\vartheta)(s-)\cdot d\nu(s) +\sum_{0<s\le t}\Delta(\kappa\vartheta)(s)\cdot\nu(s),
    \]
    where we used $\nu=P\Sigma^{-1}$ and $\Delta(\kappa\vartheta)(t)=\vartheta(t-)\Delta\kappa(t)+\kappa(t)\Delta\vartheta(t)$. This is the self-financing condition~\eqref{self-financing}. Continuity follows from $V^{\kappa\vartheta}(\cdot;\nu) =V^\vartheta(\cdot;P)/\widetilde\Sigma$. The converse direction can be shown similarly.
\end{proof}

\begin{proof}[Proof of Proposition~\ref{prop: defect of self-financeability}]
    Recalling $\mathbf{1} := (1,\dots,1) \in \mathbb{R}^d$, we can write $\varphi(\cdot) = \vartheta(\cdot) - Q^{\vartheta, \nu}(\cdot) \mathbf{1} - C^{\vartheta, \nu}(\cdot) \mathbf{1}$. First, note that for any $t > 0$
    \[
        \Delta Q^{\vartheta, \nu}(t) = \vartheta(t) \cdot \nu(t) - \vartheta(t-) \cdot \nu(t-) - \vartheta(t-) \cdot \Delta \nu(t) - \Delta \vartheta(t) \cdot \nu(t) = 0.
    \]
    Hence $\Delta \varphi(t) = \Delta \vartheta(t) - \Delta (\vartheta \cdot \nu)(t) \mathbf{1}$. Now, for the continuity of $V^\varphi(\cdot; \nu)$, we compute
    \[
        \Delta(\varphi \cdot \nu)(t) = \Delta \varphi(t) \cdot \nu(t) + \varphi(t-) \cdot \Delta \nu(t) = \Delta \vartheta(t) \cdot \nu(t) - \Delta (\vartheta \cdot \nu)(t) + \vartheta(t-) \cdot \Delta \nu(t) = 0.
    \]
    Here, we used the identity $\mathbf{1} \cdot \Delta \nu(\cdot) = 0$. To check the self-financing condition, we compute
    \begin{align*}
        (\varphi \cdot \nu)(t) - (\varphi \cdot \nu)(0) &= (\vartheta \cdot \nu)(t) - (\vartheta \cdot \nu)(0) - Q^{\vartheta, \nu}(t) \mathbf{1} \cdot \nu(t) - C^{\vartheta, \nu}(t) \mathbf{1} \cdot \nu(t) + C^{\vartheta, \nu}(0) \mathbf{1} \cdot \nu(0) \\
        &= (\vartheta \cdot \nu)(t) - (\vartheta \cdot \nu)(0) - Q^{\vartheta, \nu}(t) - \sum_{0 < s \le t} \Delta C^{\vartheta, \nu}(s) \mathbf{1} \cdot \nu(s) \\
        &= \int_0^t \vartheta(s-) \cdot d\nu(s) + \sum_{0 < s \le t} \big(\Delta \vartheta(s) - \Delta C^{\vartheta, \nu}(s) \mathbf{1}\big) \cdot \nu(s) \\
        &= \int_0^t \varphi(s-) \cdot d\nu(s) + \sum_{0 < s \le t} \Delta \varphi(s) \cdot \nu(s).
    \end{align*}
    Here, we used $\mathbf{1}\cdot d\nu \equiv 0$. Since $Q^{\vartheta,\nu}$ is continuous and predictable, and $C^{\vartheta,\nu}$ is a predictable finite variation process, $\varphi$ is predictable, RCLL, and $\nu$-integrable. Thus $\varphi$ is a trading strategy with respect to $\nu$.
\end{proof}

\bigskip

\section{Proofs of Section~\ref{sec: FGPTS}} \label{app: proofs fgpts}

\begin{proof}[Proof of Proposition~\ref{prop: additively generated portfolio theorem}]
    Since $\mathbf{1} \cdot \nu \equiv 1$, we have for $t \ge 0$
    \[
        V^\varphi(t;\nu) = \vartheta(t) \cdot \nu(t) - (\vartheta \cdot \nu)(t) - C^{\vartheta, \nu}(t) + G\big(\nu(t)\big) + \Gamma^G(t) = G\big(\nu(t)\big) + \Gamma^G(t) - C^{\vartheta, \nu}(t).
    \]
    Using the decomposition $\Gamma^G=\Gamma^{G,c}+\Gamma^{G,d}$ and the definition of $C^{\vartheta,\nu}$, we obtain
    \begin{align*}
        & \Gamma^{G,d}(t) - C^{\vartheta,\nu}(t)
        \\
        &= \sum_{0 < s \le t} \Big( -\Delta G\big(\nu(s)\big) + \vartheta(s-)\cdot\Delta\nu(s) + \Delta\vartheta(s)\cdot\nu(s) \Big) - \sum_{0 < s \le t} \Big( \vartheta(s-)\cdot\Delta\nu(s) + \Delta\vartheta(s)\cdot\nu(s) \Big) \\
        &= -\sum_{0 < s \le t} \Delta G\big(\nu(s)\big).
    \end{align*}
    Using \eqref{eq: G^c definition} yields the result \eqref{eq: additively generated portfolio}.
\end{proof}

\begin{proof} [Proof of Proposition~\ref{prop: multiplicatively generated trading strategy}]
    First, on each inter-event interval $\psi$ is continuous and adapted, hence predictable. At each $\tau_k$ with $k\ge1$, the $\mathcal{F}(\tau_k-)$-measurability of the $k$-th adjustment vector $m_k$ makes $\nu(\tau_k)$, $\alpha_k$, and $\psi(\tau_k)$ all $\mathcal{F}(\tau_k-)$-measurable, and $\tau_k$ is a predictable time, hence $\psi$ is predictable. The integrability condition follows from the continuity of $\nu$ on each inter-event interval and the local boundedness of $1/G(\nu)$, which ensures that $\eta^{(k)}$ is locally bounded. Finally, it remains to show the continuity of the value process and the self-financing condition.

    Fix $k \ge 0$, and set
    \[
        K^{(k)}(t) := \exp \bigg( \int_{\tau_k}^t \frac{d\Gamma^{G, c}(s)}{G\big(\nu(s)\big)} \bigg), \qquad t \in [\tau_k, \tau_{k+1}).
    \]
    The positivity of $G$ and local boundedness of $1/G$ ensure that $K^{(k)}$ is a continuous, positive process of finite variation on $[\tau_k, \tau_{k+1})$, with $K^{(k)}(\tau_k) = 1$ and
    \[
        dK^{(k)}(t) = K^{(k)}(t)\,\frac{d\Gamma^{G, c}(t)}{G\big(\nu(t)\big)}, \qquad t \in [\tau_k, \tau_{k+1}).
    \]
    On the open interval $(\tau_k, \tau_{k+1})$, the continuity of $\nu$ combined with It\^o's formula and the definition of $\Gamma^{G, c}$ following~\eqref{eq: decomposition of Gamma function} yields
    \[
        dG\big(\nu(t)\big) = \nabla G\big(\nu(t)\big) \cdot d\nu(t) - d\Gamma^{G, c}(t).
    \]
    Both $G(\nu)$ and $K^{(k)}$ are continuous on the interval and $K^{(k)}$ is of finite variation, so integration by parts gives
    \begin{align*}
        d\big( G(\nu(t))\, K^{(k)}(t) \big) &= G\big(\nu(t)\big)\, dK^{(k)}(t) + K^{(k)}(t)\, dG\big(\nu(t)\big) \\
        &= K^{(k)}(t)\, d\Gamma^{G, c}(t) + K^{(k)}(t) \Big( \nabla G\big(\nu(t)\big) \cdot d\nu(t) - d\Gamma^{G, c}(t) \Big) \\
        &= K^{(k)}(t)\, \nabla G\big(\nu(t)\big) \cdot d\nu(t)
        = \eta^{(k)}(t) \cdot d\nu(t),
    \end{align*}
    where the last identity uses~\eqref{eq: tps eta}.  Integrating from $\tau_k$ and recalling that $K^{(k)}(\tau_k) = 1$, we obtain
    \begin{equation} \label{eq: tps product rule}
        G\big(\nu(t)\big)\, K^{(k)}(t) = G\big(\nu(\tau_k)\big) + \int_{\tau_k}^t \eta^{(k)}(s) \cdot d\nu(s), \qquad t \in [\tau_k, \tau_{k+1}).
    \end{equation}
    
    Turning to the wealth process, the identity $\mathbf{1} \cdot \nu(t) \equiv 1$ together with~\eqref{eq: multiplicatively generated trading strategy} gives
    \[
        V^\psi(t;\nu) = \alpha_k\Big(\eta^{(k)}(t) \cdot \nu(t) - Q^{(k)}(t) - b_G\big(\nu(\tau_k)\big) \Big),
    \]
    and substituting \eqref{eq: tps Q} cancels the $\eta^{(k)}(t) \cdot \nu(t)$ term to leave
    \[
        V^\psi(t;\nu) = \alpha_k \Big( \eta^{(k)}(\tau_k) \cdot \nu(\tau_k) + \int_{\tau_k}^t \eta^{(k)}(s) \cdot d\nu(s) - b_G\big(\nu(\tau_k)\big) \Big).
    \]
    The initial condition $K^{(k)}(\tau_k) = 1$ implies $\eta^{(k)}(\tau_k) = \nabla G(\nu(\tau_k))$, and~\eqref{eq: tps balance defect} gives $\nabla G(\nu(\tau_k)) \cdot \nu(\tau_k) - b_G(\nu(\tau_k)) = G(\nu(\tau_k))$. Combining these observations with~\eqref{eq: tps product rule} identifies
    \begin{equation} \label{eq: tps interval wealth}
        V^\psi(t;\nu) = \alpha_k\, G\big(\nu(t)\big)\, K^{(k)}(t), \qquad t \in [\tau_k, \tau_{k+1}).
    \end{equation}
    
    Evaluating~\eqref{eq: tps interval wealth} at $t = \tau_k$ gives $V^\psi(\tau_k;\nu) = \alpha_k\, G(\nu(\tau_k))$. By the definition of $\alpha_k$ in~\eqref{eq: tps alpha}, this equals $V^\psi(\tau_k-;\nu)$ for every $k\ge1$. The wealth process $V^\psi$ is therefore continuous at every adjustment event, with $V^\psi(0;\nu)=G(\nu(0))$ from $\alpha_0=1$.
    
    To verify the self-financing identity~\eqref{self-financing}, fix $t \in [\tau_k, \tau_{k+1})$ for some $k \ge 0$.  Multiplying~\eqref{eq: tps product rule} by $\alpha_k$ and using~\eqref{eq: tps interval wealth} gives
    \[
        V^\psi(t;\nu) = V^\psi(\tau_k;\nu) + \int_{\tau_k}^t \alpha_k\, \eta^{(k)}(s) \cdot d\nu(s).
    \]
    On $[\tau_k,\tau_{k+1})$, the processes $\eta^{(k)}$, $Q^{(k)}$, and $\psi$ are continuous, and
    \[
        \psi(s) =\alpha_k\Big( \eta^{(k)}(s) -\big[Q^{(k)}(s)+b_G(\nu(\tau_k))\big]\mathbf1 \Big).
    \]
    Together with $\mathbf1\cdot d\nu\equiv0$, this yields
    \[
        \alpha_k\eta^{(k)}(s)\cdot d\nu(s) =\psi(s)\cdot d\nu(s), \qquad s\in(\tau_k,\tau_{k+1}).
    \]
    Hence
    \begin{equation} \label{eq: tps interval SF}
        V^\psi(t; \nu) = V^\psi(\tau_k; \nu) + \int_{\tau_k}^t \psi(s) \cdot d\nu(s).
    \end{equation}
    Taking the limit \(t\uparrow\tau_{i+1}\) in \eqref{eq: tps interval SF} for $i = 0, \dots, k-1$ and using the wealth continuity $V^\psi(\tau_{i+1};\nu)=V^\psi(\tau_{i+1}-;\nu)$, telescoping yields
    \begin{equation} \label{eq: tps SF telescoped}
        V^\psi(t;\nu) = V^\psi(0;\nu) + \sum_{i=0}^{k-1} \int_{\tau_i}^{\tau_{i+1}-} \psi(s) \cdot d\nu(s) + \int_{\tau_k}^t \psi(s) \cdot d\nu(s).
    \end{equation}
    
    For every $s$ in the open interval $(\tau_i, \tau_{i+1})$, the process $\psi$ is continuous at $s$, so $\psi(s-) = \psi(s)$. Thus, replacing $\psi(s)$ by $\psi(s-)$ leaves the integrals in~\eqref{eq: tps SF telescoped} unchanged. The standard stochastic integral $\int_0^t \psi(s-) \cdot d\nu(s)$ additionally includes the jump contribution $\psi(\tau_i-) \cdot \Delta\nu(\tau_i)$ at each adjustment event $\tau_i$, so
    \[
        \int_0^t \psi(s-) \cdot d\nu(s) = \sum_{i=0}^{k-1} \int_{\tau_i}^{\tau_{i+1}-} \psi(s) \cdot d\nu(s) + \int_{\tau_k}^t \psi(s) \cdot d\nu(s) + \sum_{i=1}^k \psi(\tau_i-) \cdot \Delta\nu(\tau_i).
    \]
    Combining with~\eqref{eq: tps SF telescoped},
    \begin{equation} \label{eq: tps SF predictable}
        V^\psi(t;\nu) = V^\psi(0;\nu) + \int_0^t \psi(s-) \cdot d\nu(s) - \sum_{i=1}^k \psi(\tau_i-) \cdot \Delta\nu(\tau_i).
    \end{equation}
    Finally, the wealth continuity $V^\psi(\tau_i;\nu) = V^\psi(\tau_i-;\nu)$ together with the algebraic Leibniz identity gives
    \[
        \psi(\tau_i-) \cdot \Delta\nu(\tau_i) + \Delta\psi(\tau_i) \cdot \nu(\tau_i) = \Delta(\psi \cdot \nu)(\tau_i) = \Delta V^\psi(\tau_i;\nu) = 0,
    \]
    so $-\sum_{i=1}^k \psi(\tau_i-) \cdot \Delta\nu(\tau_i) = \sum_{i=1}^k \Delta\psi(\tau_i) \cdot \nu(\tau_i)$.  Substituting into~\eqref{eq: tps SF predictable} gives~\eqref{self-financing}, and this proves that $\psi$ is a trading strategy with respect to $\nu$.
\end{proof}

\begin{proof}[Proof of Proposition~\ref{prop: multiplicatively generated trading strategy theorem}]
    Taking logarithms in~\eqref{eq: tps interval wealth},
    \[
        \log V^\psi(t;\nu) = \log \alpha_k + \log G \big(\nu(t)\big) + \int_{\tau_k}^t \frac{d\Gamma^{G, c}(s)}{G\big(\nu(s)\big)}, \qquad t \in [\tau_k, \tau_{k+1}).
    \]
    Using continuity of the wealth process at \(\tau_k\) and rearranging~\eqref{eq: tps interval wealth}, we obtain the recursion
    \[
        \log \alpha_k - \log \alpha_{k-1} = -\Delta\log G\big(\nu(\tau_k)\big) + \int_{\tau_{k-1}}^{\tau_k} \frac{d\Gamma^{G, c}(s)}{G\big(\nu(s)\big)}.
    \]
    Iterating this recursion with \(\log\alpha_0=0\) yields
    \[
        \log \alpha_k = -\sum_{j=1}^k \Delta\log G\big(\nu(\tau_j)\big) + \int_0^{\tau_k} \frac{d\Gamma^{G, c}(s)}{G\big(\nu(s)\big)}.
    \]
    Substituting this expression into the per-interval logarithm above and collecting the integrals over $[0, \tau_k]$ and $[\tau_k, t]$ produces
    \[
        \log V^\psi(t;\nu)
        = \log G(\nu(t)) - \sum_{0 < s \le t} \Delta\log G(\nu(s))
        + \int_0^t \frac{d\Gamma^{G, c}(s)}{G(\nu(s))}
    \]
    for $t \in [\tau_k, \tau_{k+1})$.  The first two terms are precisely $(\log G)^c(\nu(t))$, and~\eqref{eq: value process of multiplicatively generated trading strategy} follows.
\end{proof}

\begin{proof}[Proof of Corollary~\ref{coro: representation of multiplicatively generated trading strategy balanced case}]
    Fix $t \in [\tau_k, \tau_{k+1})$ for some $k \ge 0$. Write $c_k:=b_G(\nu(\tau_k))$, and let $K^{(k)}$ be the exponential factor defined in the proof of Proposition~\ref{prop: multiplicatively generated trading strategy}. Substituting \eqref{eq: multiplicatively generated trading strategy} and \eqref{eq: tps interval wealth} into \eqref{def:pi} and using \eqref{eq: tps eta} yield
    \begin{equation} \label{eq: tps relport pre-cancel}
        \pi^\psi_i(t)
        = \frac{\nu_i(t)\big(\eta^{(k)}_i(t) - Q^{(k)}(t) - c_k\big)}{G\big(\nu(t)\big)\, K^{(k)}(t)}
        = \frac{\nu_i(t)\,\partial_i G\big(\nu(t)\big)}{G\big(\nu(t)\big)} - \frac{\nu_i(t)\big(Q^{(k)}(t) + c_k\big)}{G\big(\nu(t)\big)\, K^{(k)}(t)}.
    \end{equation}
    The identity $\alpha_k(\eta^{(k)}(t) \cdot \nu(t) - Q^{(k)}(t) - c_k) = V^\psi(t;\nu) = \alpha_k\, G(\nu(t))\, K^{(k)}(t)$, along with $\eta^{(k)}(t) \cdot \nu(t) = K^{(k)}(t)\sum_{j=1}^d \nu_j(t)\, \partial_j G(\nu(t))$, gives
    \[
        \frac{Q^{(k)}(t) + c_k}{K^{(k)}(t)} = \sum_{j=1}^d \nu_j(t)\,\partial_j G\big(\nu(t)\big) - G\big(\nu(t)\big).
    \]
    Substituting this into \eqref{eq: tps relport pre-cancel} proves \eqref{eq: tps related portfolio} for every $t \in [0, \infty)$ since the argument is independent of $k \ge 0$. The balanced simplification follows from $\sum_j x_j\,\partial_j G(x) = G(x)$.
\end{proof}

\begin{proof} [Proof of Proposition~\ref{prop: change of basis}]
    The identity~\eqref{eq: change of basis wealth} is immediate from the relationship \eqref{def:P}. First, we assume (i); then the continuity of $V^{\vartheta \odot N^{-1}}(\cdot;S)$ follows by \eqref{eq: change of basis wealth}. Applying Itô’s product rule to \(P_i=S_i/N_i\), the dynamics of $P_i$ can be written as
    \begin{align*}
        dP_i(t) &= \frac{1}{N_i(t-)}\,dS_i(t) + S_i(t-)\,d\Big(\frac{1}{N_i}\Big)(t) + d\Big[S_i,\frac{1}{N_i}\Big](t) \\
        &= \frac{1}{N_i(t-)}\,dS_i(t) + S_i(t)\,d\Big(\frac{1}{N_i}\Big)(t),
    \end{align*}
    where the covariation term vanishes because $S_i$ is continuous and $N_i^{-1}$ is of finite variation. Since $N_i^{-1}$ is constant between adjustment events, this identity can also be written as
    \[
        P_i(t)-P_i(0) = \int_0^t\frac{dS_i(s)}{N_i(s-)} + \sum_{0<s\le t}S_i(s)\Delta\Big(\frac{1}{N_i}\Big)(s).
    \]
    Using this with the continuity of $S$ and $\vartheta \cdot P$, we compute for $t \ge 0$
    \begin{align*}
        & V^{\vartheta \odot N^{-1}}(t; S) - V^{\vartheta \odot N^{-1}}(0; S) = V^{\vartheta}(t; P) - V^{\vartheta}(0; P) = \int_0^t \vartheta(s-) \cdot dP(s) + \sum\limits_{0 < s \le t} \Delta \vartheta(s) \cdot P(s) \\
        &\qquad = \int_0^t \sum_{i=1}^d \frac{\vartheta_i(s-)}{N_i(s-)} dS_i(s) + \sum_{0 < s \le t} \sum_{i=1}^d \vartheta_i(s-) S_i(s) \Delta\Big( \frac{1}{N_i} \Big)(s) + \sum_{0 < s \le t} \sum_{i=1}^d \Delta \vartheta_i(s) P_i(s) \\
        &\qquad = \int_0^t \sum_{i=1}^d \frac{\vartheta_i(s-)}{N_i(s-)} dS_i(s) + \sum_{0 < s \le t} \sum_{i=1}^d \big( \vartheta_i(s-) \Delta P_i(s) + \Delta \vartheta_i(s) P_i(s) \big) 
        = \int_0^t \sum_{i=1}^d \frac{\vartheta_i(s-)}{N_i(s-)} dS_i(s).
    \end{align*}
    Again by the continuity of $S$ and $\vartheta \cdot P$, the jump correction term in \eqref{self-financing} for the strategy $\vartheta \odot N^{-1}$ vanishes:
    \begin{align*}
        \sum_{0 < s \le t} \sum_{i=1}^d \Delta \Big( \frac{\vartheta_i(s)}{N_i(s)} \Big) S_i(s) &= \sum_{0 < s \le t} \sum_{i=1}^d \Big( \frac{\vartheta_i(s)}{N_i(s)}S_i(s) - \frac{\vartheta_i(s-)}{N_i(s-)}S_i(s) \Big) \\
        &= \sum_{0 < s \le t} \sum_{i=1}^d \Big( \frac{\vartheta_i(s)}{N_i(s)}S_i(s) - \frac{\vartheta_i(s-)}{N_i(s-)}S_i(s-) \Big) = \sum_{0 < s \le t} \Delta(\vartheta \cdot P)(s) = 0.
    \end{align*}
    This proves (i) $\Rightarrow$ (ii). The other implication follows by a similar argument in the reverse direction.
\end{proof}

\bigskip

\noindent \textbf{Funding}

\medskip

\noindent Jongbong An is supported by the National Research Foundation of Korea under Grants RS-2026-25475614 and RS-2025-00513609, funded by the Korea government (MSIT). Donghan Kim is supported by the National Research Foundation of Korea under Grant RS-2025-00513609, and by an Ewon Assistant Professorship Award at KAIST. 

\bigskip

\bibliographystyle{plain}
\bibliography{bibliography}

@article{strong2014generalizations,
  author  = {Strong, Winslow},
  title   = {Generalizations of functionally generated portfolios with applications to statistical arbitrage},
  journal = {SIAM Journal on Financial Mathematics},
  year    = {2014},
  volume  = {5},
  number  = {1},
  pages   = {472--492},
  doi     = {10.1137/130907458}
}

@article{rufxie2019generalised,
  author  = {Ruf, Johannes and Xie, Kangjianan},
  title   = {Generalised {Lyapunov} functions and functionally generated trading strategies},
  journal = {Applied Mathematical Finance},
  year    = {2019},
  volume  = {26},
  number  = {4},
  pages   = {293--327},
  doi     = {10.1080/1350486X.2019.1584041}
}

@article{palwong2016geometry,
  author  = {Pal, Soumik and Wong, Ting-Kam Leonard},
  title   = {The geometry of relative arbitrage},
  journal = {Mathematics and Financial Economics},
  year    = {2016},
  volume  = {10},
  number  = {3},
  pages   = {263--293},
  doi     = {10.1007/s11579-015-0159-z}
}

@article{kardarasrobertson2021ergodic,
  author  = {Kardaras, Constantinos and Robertson, Scott},
  title   = {Ergodic robust maximization of asymptotic growth},
  journal = {The Annals of Applied Probability},
  year    = {2021},
  volume  = {31},
  number  = {4},
  pages   = {1787--1819},
  doi     = {10.1214/20-AAP1634}
}

@article{larssonruf2021relative,
  author  = {Larsson, Martin and Ruf, Johannes},
  title   = {Relative arbitrage: Sharp time horizons and motion by curvature},
  journal = {Mathematical Finance},
  year    = {2021},
  volume  = {31},
  number  = {3},
  pages   = {885--906},
  doi     = {10.1111/mafi.12303}
}

@misc{laishkolnikovsoner2025relative,
  author        = {Lai, Jou-Hua and Shkolnikov, Mykhaylo and Soner, H. Mete},
  title         = {Relative arbitrage problem under eigenvalue lower bounds},
  year          = {2025},
  eprint        = {2512.17702},
  archivePrefix = {arXiv},
  primaryClass  = {q-fin.MF},
  note          = {arXiv:2512.17702},
  url           = {https://arxiv.org/abs/2512.17702}
}

@article{Portfolio:Optimization:SPT,
author = {Campbell, Steven and Wong, Ting-Kam Leonard},
title = {Functional Portfolio Optimization in Stochastic Portfolio Theory},
journal = {SIAM Journal on Financial Mathematics},
volume = {13},
number = {2},
pages = {576-618},
year = {2022},
doi = {10.1137/21M1417715},
URL = {https://doi.org/10.1137/21M1417715},
eprint = {https://doi.org/10.1137/21M1417715}
}

@article{Option:pricing:SPT,
  author  = {Fernholz, Ricardo T. and Fernholz, Robert},
  title   = {Portfolios generated by contingent claim functions, with applications to option pricing},
  journal = {Finance and Stochastics},
  year    = {2026},
  volume  = {30},
  issue = {3},
  pages   = {873--901},
  doi     = {10.1007/s00780-026-00597-5}
}

@article{Market-to-book,
  author  = {Kim, Donghan},
  title   = {Market-to-book ratio in stochastic portfolio theory},
  journal = {Finance and Stochastics},
  year    = {2023},
  volume  = {27},
  issue = {2},
  pages   = {401--434},
  doi     = {10.1007/s00780-023-00501-5}
}

@misc{price:impact:SPT,
      title={Stochastic portfolio theory with price impact}, 
      author={David Itkin},
      year={2026},
      eprint={2506.07993},
      note={arXiv:2506.07993},
      primaryClass={q-fin.MF},
      url={https://arxiv.org/abs/2506.07993}, 
}

@article{Open:market:Jacobi,
author = {David Itkin and Martin Larsson},
title = {{Open markets and hybrid Jacobi processes}},
volume = {34},
journal = {The Annals of Applied Probability},
number = {3},
publisher = {Institute of Mathematical Statistics},
pages = {2940 -- 2985},
year = {2024},
doi = {10.1214/23-AAP2030},
URL = {https://doi.org/10.1214/23-AAP2030}
}

@article{Ruf:Xie,
author = {Ruf, Johannes and Xie, Kangjianan},
title = {The Impact of Proportional Transaction Costs on Systematically Generated Portfolios},
journal = {SIAM Journal on Financial Mathematics},
volume = {11},
number = {3},
pages = {881-896},
year = {2020},
doi = {10.1137/19M1282313},
URL = {https://doi.org/10.1137/19M1282313},
eprint = {https://doi.org/10.1137/19M1282313}
}

@article{Universal:SPT,
author = {Cuchiero, Christa and Schachermayer, Walter and Wong, Ting-Kam Leonard},
title = {Cover's universal portfolio, stochastic portfolio theory, and the numéraire portfolio},
journal = {Mathematical Finance},
volume = {29},
number = {3},
pages = {773-803},
doi = {https://doi.org/10.1111/mafi.12201},
url = {https://onlinelibrary.wiley.com/doi/abs/10.1111/mafi.12201},
eprint = {https://onlinelibrary.wiley.com/doi/pdf/10.1111/mafi.12201},
year = {2019}
}

@article{rough:path:SPT,
author = {Allan, Andrew L. and Cuchiero, Christa and Liu, Chong and Prömel, David J.},
title = {Model-free portfolio theory: A rough path approach},
journal = {Mathematical Finance},
volume = {33},
number = {3},
pages = {709-765},
doi = {https://doi.org/10.1111/mafi.12376},
url = {https://onlinelibrary.wiley.com/doi/abs/10.1111/mafi.12376},
eprint = {https://onlinelibrary.wiley.com/doi/pdf/10.1111/mafi.12376},
year = {2023}
}

@article{vol:arb,
author = {E. Robert Fernholz and Ioannis Karatzas and Johannes Ruf},
title = {{Volatility and arbitrage}},
volume = {28},
journal = {The Annals of Applied Probability},
number = {1},
publisher = {Institute of Mathematical Statistics},
pages = {378 -- 417},
year = {2018},
doi = {10.1214/17-AAP1308},
URL = {https://doi.org/10.1214/17-AAP1308}
}

@article{Schied:model-free,
  author  = {Schied, Alexander and Speiser, Leo and Voloshchenko, Iryna},
  title   = {Model-free portfolio theory and its functional master formula},
  journal = {SIAM Journal on Financial Mathematics},
  year    = {2018},
  volume  = {9},
  pages   = {1074--1101},
  doi     = {10.1137/16M1079828}
}

@article{Karatzas:Kim,
  author  = {Karatzas, Ioannis and Kim, Donghan},
  title   = {Trading strategies generated pathwise by functions of market weights},
  journal = {Finance and Stochastics},
  year    = {2020},
  volume  = {24},
  pages   = {423--463},
  doi     = {10.1007/s00780-019-00414-2}
}

@article{Karatzas:Kim2,
	author = {Karatzas, Ioannis and Kim, Donghan},
	title = {Open markets},
	journal = {Mathematical Finance},
	volume = {31},
	number = {4},
	pages = {1111-1161},
	doi = {https://doi.org/10.1111/mafi.12294},
	url = {https://onlinelibrary.wiley.com/doi/abs/10.1111/mafi.12294},
	eprint = {https://onlinelibrary.wiley.com/doi/pdf/10.1111/mafi.12294},
	year = {2021}
}

@INCOLLECTION{FK_survey,
	author = {Robert Fernholz and Ioannis Karatzas},
	title = {Stochastic {P}ortfolio {T}heory: an overview},
	booktitle = {Handbook of Numerical Analysis},
	publisher = {Elsevier},
	year = {2009},
	volume = {15},
	pages = {89-167},
	doi = {10.1016/S1570-8659(08)00003-3}
}

@article{bayraktar2024quantifying,
  title={Quantifying dimensional change in stochastic portfolio theory},
  author={Bayraktar, Erhan and Kim, Donghan and Tilva, Abhishek},
  journal={Mathematical Finance},
  volume={34},
  number={3},
  pages={977--1021},
  year={2024},
  publisher={Wiley Online Library}
}

@book{fernholz2002stochastic,
  title={Stochastic Portfolio Theory},
  author={Fernholz, E. Robert},
  series={Applications of Mathematics},
  volume={48},
  publisher={Springer-Verlag},
  address={New York},
  year={2002}
}

@article{karatzas2017trading,
  title={{Trading strategies generated by Lyapunov functions}},
  author={Karatzas, Ioannis and Ruf, Johannes},
  journal={Finance and Stochastics},
  volume={21},
  number={3},
  pages={753--787},
  year={2017},
  publisher={Springer}
}

@misc{nikkei225indexguidebook,
  author = {Nikkei Inc.},
  title = {Nikkei Stock Average Index Guidebook},
  year = {2025},
  howpublished = {\url{https://indexes.nikkei.co.jp/nkave/archives/file/nikkei_stock_average_guidebook_en.pdf}},
  note = {Updated: 2025-07-24}
}

@misc{spdji_djia_methodology,
  author = {{S\&P Dow Jones Indices}},
  title = {{Dow Jones Averages Methodology}},
  year = {2026},
  howpublished = {\url{https://www.spglobal.com/spdji/en/documents/methodologies/methodology-dj-averages.pdf}},
  note = {S\&P Global, methodology document for the Dow Jones Industrial, Transportation, and Utility Averages (incl.\ divisor mechanism). Updated: 2026-02}
}

@article{strong2011diversity,
  title={Diversity and arbitrage in a regulatory breakup model},
  author={Strong, Winslow and Fouque, Jean-Pierre},
  journal={Annals of Finance},
  volume={7},
  number={3},
  pages={349--374},
  year={2011},
  publisher={Springer}
}

@article{damjanovic2025price,
  title={Price vs. market-cap-weighted portfolio diversification: does it matter?},
  author={Damjanovic, Aleksandar and Drenovak, Mikica and Uro{\v{s}}evi{\'c}, Branko and Jelic, Ranko},
  journal={The European Journal of Finance},
  volume={32},
  number={4--6},
  pages={564--585},
  year={2026},
  doi={10.1080/1351847X.2025.2529216},
}

@book{karatzas1991brownian,
  author    = {Karatzas, Ioannis and Shreve, Steven E.},
  title     = {Brownian Motion and Stochastic Calculus},
  edition   = {2nd},
  series    = {Graduate Texts in Mathematics},
  volume    = {113},
  publisher = {Springer-Verlag},
  address   = {New York},
  year      = {1991}
}

@article{cuchiero2025signature,
  title={Signature methods in stochastic portfolio theory},
  author={Cuchiero, Christa and M{\"o}ller, Janka},
  journal={SIAM Journal on Financial Mathematics},
  volume={16},
  number={4},
  pages={1239--1303},
  year={2025},
  publisher={SIAM}
}

@article{campbell2025macroscopic,
  title={Macroscopic properties of equity markets: stylized facts and portfolio performance},
  author={Campbell, Steven and Song, Qien and Wong, Ting-Kam Leonard},
  journal={Quantitative Finance},
  volume={25},
  number={9},
  pages={1375--1397},
  year={2025},
  publisher={Taylor \& Francis}
}

@article{fernholz2022permutation,
  title={Permutation-weighted portfolios and the efficiency of commodity futures markets},
  author={Fernholz, Ricardo T and Fernholz, Robert},
  journal={Annals of Finance},
  volume={18},
  number={1},
  pages={81--108},
  year={2022},
  publisher={Springer}
}

@article{banner2019diversification,
  author  = {Banner, Adrian and Fernholz, E. Robert and Papathanakos, Vassilios and Ruf, Johannes and Schofield, David},
  title   = {Diversification, Volatility, and Surprising Alpha},
  journal = {Journal of Investment Consulting},
  volume  = {19},
  number  = {1},
  pages   = {23--30},
  year    = {2019},
  url     = {https://ssrn.com/abstract=3522675}
}

\end{document}